\documentclass[]{article}
\usepackage{t1enc}
\usepackage[cp1250]{inputenc}
\usepackage{newcent,latexsym,amssymb,amsthm,amsmath,graphics,color}
\usepackage{epsfig}
\usepackage[T1]{fontenc}
\usepackage{amsthm,amsmath}
\usepackage{enumerate}
\usepackage{tikz}

\usetikzlibrary{calc}
\usetikzlibrary{decorations.pathreplacing}
\usetikzlibrary{decorations.markings}

\tikzset{->-/.style={decoration={
  markings,
  mark=at position #1 with {\arrow{>}}},postaction={decorate}}}

\theoremstyle{plain}
\newtheorem{theorem}{Theorem}
\newtheorem{lemma}[theorem]{Lemma}

\theoremstyle{definition}

\newif\ifdraftpaper
 \draftpaperfalse
\newcommand{\ignore}[1]{}

\newcommand{\FC}{\mbox{\rm FC}}

\newcommand{\bB}{\mathbf{B}}% partitions of a set
\newcommand{\bC}{\mathbf{C}}% sets of blocks
\newcommand{\bE}{\mathbf{E}}% sets of blocks
\newcommand{\bbB}{\mathbb{B}}% factorization of tpo
\newcommand{\fA}{\mathfrak{A}}% structure (tpo)
\newcommand{\fB}{\mathfrak{B}}% structure (tpo)

\newcommand{\cL}{\mathcal{L}}% logics
\newcommand{\bA}{\mathbf{A}}% Matrices 
\newcommand{\fC}{\mathfrak{C}}% clicque
\newcommand{\fD}{\mathfrak{D}}% diatom

\newcommand{\ft}{\mathfrak{t}}% rather silly vector

\renewcommand{\phi}{\varphi} % Nicer-looking phi

\newcommand{\FinSat}{\ensuremath{\textit{FinSat}}}

\newcommand{\tp}{\ensuremath{\mbox{\rm tp}}}

\newcommand{\FOT}{\ensuremath{\cL^2}}
\newcommand{\LtkE}{\ensuremath{\mbox{$\mathcal{L}^2k\mbox{E}$}}}
\newcommand{\LtoE}{\ensuremath{\mbox{$\mathcal{L}^21\mbox{E}$}}}
\newcommand{\LttE}{\ensuremath{\mbox{$\mathcal{L}^22\mbox{E}$}}}

\newcommand{\LtkT}{\ensuremath{\mbox{$\mathcal{L}^2k\mbox{T}$}}}
\newcommand{\LtoT}{\ensuremath{\mbox{$\mathcal{L}^21\mbox{T}$}}}

\newcommand{\LtoPOu}{\ensuremath{\mbox{$\mathcal{L}^21\mbox{\rm PO}^u$}}}
\newcommand{\LtoPO}{\ensuremath{\mbox{$\mathcal{L}^21\mbox{\rm PO}$}}}

\newcommand{\ExpSpace}{\textsc{ExpSpace}}
\newcommand{\NExpTime}{\textsc{NExpTime}}
\newcommand{\TwoExpTime}{2\textsc{-ExpTime}}
\newcommand{\TwoNExpTime}{2\textsc{-NExpTime}}
\newcommand{\ThreeNExpTime}{3\textsc{-NExpTime}}

\newcommand{\set}[1]{\{#1\}}

\newcommand{\sizeOf}[1]{\lVert #1 \rVert}

\newcommand{\modt}[1]{{\lfloor#1\rfloor}}

\newcommand{\cutout}[1]{}

\graphicspath{{5-1-Graphics/}{/}}

\newcommand{\leftCell}{L}
\newcommand{\rightCell}{R}
\newcommand{\converseDiatom}{I}
\newcommand{\relVar}{\mathfrak{s}}

\newcommand{\ianIgnore}[1]{}

\title{The Finite Satisfiability Problem for Two-Variable, First-Order Logic with one Transitive Relation is Decidable}
\author{Ian Pratt-Hartmann\\
\ \\
{\small
School of Computer Science, University of Manchester}\\
{\small Instytut Matematyki i Informatyki, Uniwersytet Opolski}}
\date{}
\begin{document}
\maketitle
\section{Introduction}
\label{sec:intro}

The {\em two-variable fragment}, 
henceforth denoted $\FOT$, is the fragment of first-order logic with equality but without function-symbols, in which only two logical variables may appear. It is well-known that $\FOT$ has the finite model property, and that its satisfiability (= finite satisfiability) problem is\linebreak
$\NExpTime$-complete~\cite{GKV97}. It follows that it is impossible, within $\FOT$, to express the condition that a given binary predicate $r$ denotes a transitive relation, since in that case the  $\FOT$-formula $\forall x \neg r(x,x) \wedge \forall x \exists y r(x,y)$ becomes an axiom of infinity. This observation has prompted investigation of what happens when $\FOT$ is enriched by imposing various semantic restrictions on the interpretations of certain predicates. For $k >0$, denote by
$\LtkT$ the logic whose formulas are exactly those of $\FOT$, but where $k$ distinguished predicates are required to be interpreted as \textit{transitive relations}, and denote by
$\LtkE$ the same set of formulas, but where $k$ distinguished predicates are required to be interpreted as 
\textit{equivalence relations}. For each of these logics, the question arises as to whether the satisfiability and finite satisfiability problems are decidable, and, if so, what their computational complexity is.

The following is known. 
(i) $\LtoE$ has the finite model property, and its satisfiability (= finite satisfiability) problem is $\NExpTime$-complete~\cite{fo21t:KO05}. (ii) $\LttE$ lacks the finite model property, but its satisfiability and finite satisfiability problems are both $\TwoNExpTime$-complete~\cite{fo21t:kmp-ht}.
(iii) For $k \geq 3$, the satisfiability and finite satisfiability problems for $\LtkE$ are both undecidable~\cite{fo21t:KO05}. 
(iv) $\LtoT$ lacks the finite model property and its satisfiability problem is \TwoExpTime-hard and in \TwoNExpTime~\cite{KT09}. (v) For $k \geq 2$, the satisfiability and finite satisfiability problems for $\LtkT$ are both undecidable~\cite{fo21t:Kie05}. (In fact, the satisfiability and finite satisfiability problems for the two-variable
fragment with one transitive relation and one equivalence relation are already undecidable~\cite{fo21t:KT09}.)
This resolves the decidability and (within narrow limits) the complexity of the satisfiability and finite satisfiability problems for all of the logics
$\LtkT$ and $\LtkE$ except for one case: the finite satisfiability problem for $\LtoT$, where decidability is currently open. This article deals with that case, by showing that the finite satisfiability problem for $\LtoT$ is in \ThreeNExpTime. The best currently known lower bound for this problem is $\TwoExpTime$-hard~\cite{fo21t:Kie03}. We remark that the approach employed in~\cite{KT09} to establish
the decidability of the satisfiability problem for $\LtoT$ breaks down if models are required to be finite: the
algorithm presented here for determining finite satisfiability employs a quite different strategy.

Denote by $\LtoPO$ the logic defined in exactly the same way as $\LtoT$, except that the distinguished 
binary relation is constrained to be interpreted as a (strict) {\em partial order}---i.e.~as an transitive and irreflexive relation. Since the $\FOT$-formula $\forall x \forall y \neg r(x,x)$ asserts that $r$ is irreflexive, it follows that $\LtoPO$ no stronger, in terms of expressive power, than $\LtoT$. In addition, we take the logic $\LtoPOu$ to be the
fragment of $\LtoPO$ in which---apart from equality and the distinguished (partial order) predicate---only \textit{unary} predicates are allowed. Our strategy in the sequel is first to consider $\LtoPOu$. Structures 
interpreting this logic are, in effect, partial orders in which each element is assigned one of a finite number of {\em types}. We obtain a $\TwoNExpTime$ upper complexity-bound on the finite satisfiability problem for this logic, by introducing a method for `factorizing' such typed partial orders into smaller partial orders on blocks of elements of the same type. We then extend this upper bound
to $\LtoPO$ by exhibiting a method to eliminate all binary predicates in $\LtoPO$-formulas (other than equality and the distinguished predicate). Finally, we obtain the $\ThreeNExpTime$ upper complexity-bound on the finite satisfiability problem for $\LtoT$ by 
exhibiting a method to replace the distinguished transitive relation by a partial order.
This latter reduction produces an exponential increase in the size of the formula in question.
 
Stronger complexity-theoretic upper bounds 
are available when the distinguished predicates are required to be interpreted as
\textit{linear} orders:
the satisfiability and finite satisfiability problems for $\FOT$ together with one linear order are both \NExpTime-complete~\cite{fo21t:Otto01};
the finite satisfiability problem for $\FOT$ together with two linear orders (and only unary non-navigational predicates) is \ExpSpace-complete~\cite{fo21t:SchZ10}; with three linear orders, satisfiability and finite satisfiability are both undecidable~\cite{fo21t:Kie2011,fo21t:Otto01}. Also somewhat related
to $\LtoPOu$ is the propositional modal logic 
known as {\em navigational XPATH}, which features a signature of proposition letters interpreted over vertices of some finite, ordered tree, together with modal operators giving access to vertices standing in the relations of {\em daughter} and {\em next-sister}, as well as their transitive closures. It is known that, over finite trees, navigational XPATH has the same expressive power as
two-variable, first-order logic with a signature consisting of unary predicates (representing properties of vertices) together with binary `navigational' predicates
(representing the modal accessibility relations). The exact complexity of satisfiability for all natural variants of this logic is given in~\cite{fo21t:bbcklmw}.

To convey a sense of the expressive power of the logics we are working with, we give an example showing that the logic $\LtoPOu$ can force the existence of an infinite anti-chain: that is, an infinite collection of elements none of which is related to any other in the partial ordering. The example is due to E.~Kiero\'{n}ski (personal communication). In the following, we use $<$ as the distinguished binary predicate of $\LtoPOu$ (written using infix notation).
First of all, the formulas 
\begin{align*}
& \exists x . p(x) & & \forall x \forall y (p(x) \wedge p(y) \rightarrow (x< y \vee  x = y  \vee y < x))
\end{align*}
ensure that elements satisfying $p$ form a non-empty linear order. Pick some such element $a_1$. Now the formulas
\begin{align*}
& \forall x (p(x) \rightarrow \exists y (\neg x< y \wedge \neg y < x \wedge q(y)) & & 
  \forall x (q(x) \rightarrow \exists y (x< y \wedge p(y))
\end{align*}
ensure that, for every element, say $a_i$, satisfying $p$, there is an incomparable element, say $b_i$, satisfying $q$,
and, for every element $b_i$ satisfying $q$, there is a greater element, say $a_{i+1}$, satisfying $p$. Thus,
we generate sequences of elements $a_1, a_2, \dots, $ satisfying $p$, and $b_1, b_2, \dots, $ satisfying $q$.
A moment's thought shows that for all $i$, $a_i < a_{i+1}$, so that, by a simple 
induction, $a_i < b_j$ for all $i <j$. This immediately implies that the $b_j$ are all distinct, since $a_i$ and $b_i$ 
are, by construction,  incomparable. The formula
\begin{align*}
& \forall x \forall y (q(x) \wedge q(y) \rightarrow (\neg x< y \wedge \neg y < x)).
\end{align*}
Then secures the sought-after infinite anti-chain. That the formulas are satisfiable is shown by the partially-ordered structure
depicted in Fig.~\ref{fig:example}. We remark that, even under the assumption that structures are finite, $\LtoPOu$ can force doubly-exponential-sized models; this is demonstrated, for example, using the  construction of~\cite{fo21t:Kie03}.
\begin{figure}
\begin{center}
\begin{tikzpicture}
\filldraw (2,0) circle (0.05);
\draw[->-=0.6] (2,0) -- (4,0);

\filldraw (4,0) circle (0.05);
\draw[->-=0.6] (4,0) -- (6,0);

\filldraw (6,0) circle (0.05);
\draw[->-=0.6] (6,0) -- (8,0);

\filldraw (8,0) circle (0.05);

\filldraw (2,1) circle (0.05);
\draw[->-=0.6] (2,1) -- (4,0);
\draw[->-=0.6,dashed] (8,0) -- (9,0);

\filldraw (4,1) circle (0.05);
\draw[->-=0.6] (4,1) -- (6,0);

\filldraw (6,1) circle (0.05);
\draw[->-=0.6] (6,1) -- (8,0);

\filldraw (8,1) circle (0.05);
\draw[->-=0.6,dashed] (8,1) -- (9,0.5);

\filldraw[opacity=0.1] (1.75,-0.25) rectangle (10,0.25);
\filldraw[opacity=0.05] (1.75,0.75) rectangle (10,1.25);

\coordinate [label=center:{$p$}] (a) at (1.5,0);
\coordinate [label=center:{$q$}] (a) at (1.5,1);

\end{tikzpicture}
\end{center}
\caption{A linear order on the elements satisfying $p$, and an anti-chain on the elements satisfying $q$.} 
\label{fig:example}
\end{figure}
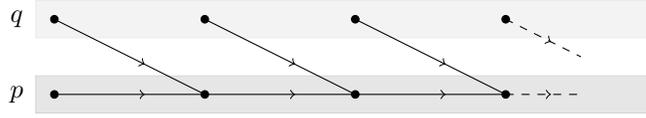

\section{Preliminaries}
\label{subsec:prelims}
We employ standard model-theoretic notation: structures are indicated by (possibly decorated) \textit{fraktur} letters $\fA$, $\fB$, \dots, and their domains by the corresponding Roman letters $A$, $B$, \dots \ .   In this paper, we adopt the non-standard assumption that all structures have cardinality at least 2. Thus, the formula $\forall x \exists y  (x \neq y)$ is for us a validity, and $\forall x \forall y  (x = y)$ a contradiction. 
This assumption does not represent a significant restriction:
over domains of size 1, first-order logic reduces to  propositional logic. 

A binary relation $R$ on some carrier set $A$ is \textit{transitive} if $aRb$ and $bRc$ implies $aRc$, \textit{reflexive} if $aRa$ always holds, \textit{irreflexive} if $a {R} a$ never holds, and {\em anti-symmetric} if $aRb$ and $bRa$ implies $a=b$. Every transitive, irreflexive relation is trivially anti-symmetric. A \textit{weak partial order} is a relation that is transitive, reflexive and anti-symmetric; a \textit{strict partial order} is a relation that is transitive and irreflexive. If $R$ is a weak partial order and $I$ the identity (diagonal) relation on $A$, then $R \setminus I$ is a strict partial order; moreover, all strict partial orders on $A$ arise in this way. Likewsie, if $R$ is a strict partial order then $R \cup I$ is a weak partial order; moreover, all weak partial orders on $A$ arise in this way. In the sequel, the unmodified phrase {\em partial order} will always mean \textit{strict} partial order. 

The \textit{two-variable fragment}, here denoted $\FOT$, is the fragment of first-order logic with equality but without function-symbols, in which only two variables, $x$ and $y$, may appear. There are no other syntactic restrictions. In particular, formulas such as $\forall x (p(x) \rightarrow \exists y (r(x,y) \wedge \exists x . s(y,x)))$, in which bound occurrences of a variable $u$ may appear within the scope of a quantifier $Qu$, are allowed. It is routine to show that predicates having arity other than 1 or 2 add no effective expressive power in the context of $\FOT$. It is likewise routine to show that individual constants add no effective expressive power
given the presence of the equality predicate. Henceforth, then, we shall take all signatures to consist only of unary and binary predicates. 

We define $\LtoT$ to be the set of formulas of $\FOT$ over any signature of unary and binary predicates which features a distinguished binary predicate $\ft$. The semantics of $\LtoT$ is exactly as for $\FOT$, except that the interpretation of $\ft$ is required to be a \textit{transitive relation}. Similarly, we define $\LtoPO$ to be the set of formulas of $\FOT$ over any signature of unary and binary predicates which features a distinguished binary predicate $<$ (written using infix notation). The semantics of $\LtoPO$ is exactly as for $\FOT$, except that the interpretation of $<$ is required to be a \textit{partial order}. Finally, we define $\LtoPOu$ to be the subset of $\LtoPO$ in which no binary predicates other than $=$ and $<$ appear. 

A formula of $\FOT$ is said to be {\em unary} if it features just one free variable. A unary formula $\zeta$ is generally silently assumed to have $x$ as its only free variable; if $\zeta$ is such a formula, we write
$\zeta(y)$ for the result of replacing $x$ in $\zeta$ by $y$. 
The unary formulas $\mu_1, \dots, \mu_n$ are \textit{mutually exclusive} if $\models \forall x (\mu_i \rightarrow \neg \mu_j)$ for all $i$ ($1 \leq i < j \leq n$). 
Any formula $\eta$ of $\FOT$ with two free variables is assumed to have those variables taken in the order $x,y$. Thus, we write
$\fA \models \eta[a,b]$, where $a$, $b$ are elements of $A$, to indicate that $\eta$ is satisfied in 
$\fA$ under the assignment $a \mapsto x$ and $b \mapsto y$. For the purposes of this paper, we may take the {\em size} of an $\FOT$-formula $\phi$, denoted $\sizeOf{\phi}$, to be the number of symbols it contains.

For any signature $\sigma$, a $\sigma$-\textit{atom} is a formula of the form $p(\bar{x})$ where $p$ is a predicate of $\sigma$ and $\bar{x}$ a tuple of variables of the appropriate arity. A $\sigma$-\textit{literal} is a $\sigma$-atom or a negated $\sigma$-atom. The reference to $\sigma$ is omitted if unimportant or clear from context.
A \textit{clause} is a disjunction of literals; we allow $\bot$ (the disjunction of no literals) to be a clause. A \textit{1-type} over $\sigma$ is a maximal consistent set of equality-free
$\sigma$-literals involving only the variable $x$; and a \textit{2-type} over $\sigma$ is a maximal consistent set of $\sigma$-literals involving the variables $x$ and $y$. (Thus, 1- and 2-types are what are sometimes called \textit{atomic} 1- and 2-types.) Consistency here is to be understood as taking into account the semantic constraints on distinguished predicates. Thus, if $\sigma$ contains $<$, then the 1-type over $\sigma$ contains the literals $x = x$ and $\neg x < x$, with similar restrictions applying to 2-types. Likewise, if $\sigma$ contains $\ft$, then any 2-type
containing the literals $\ft(x,y)$ and $\ft(y,x)$ also contains $\ft(x,x)$ and $\ft(y,y)$. We
usually identify 1- and 2-types with the conjunction of their literals.
If $\fA$ is a structure and $a, b \in A$, we write $\tp^\fA[a]$ for the unique 1-type satisfied in $\fA$ by $a$ and 
$\tp^\fA[a,b]$ for the unique 2-type satisfied by $\langle a, b \rangle$.

A formula $\phi$ of $\FOT$ (or of $\LtoPO$ or $\LtoT$) 
is said to be in {\em standard normal form} if it conforms to the pattern
\begin{equation}
\begin{split}
 \forall x \forall y (x = y \vee \eta) \wedge
       \bigwedge_{h=0}^{m-1} \forall x \exists y (x \neq y \wedge \theta_h), 
\end{split}
\label{eq:snf}
\end{equation}
where $\eta$,  $\theta_0, \dots,\theta_{m-1}$ are quantifier- and equality-free formulas, with 
$m \geq 1$. A formula is said to be in {\em weak normal form} if it conforms to the pattern
\begin{equation}
\begin{split}
\bigwedge_{\zeta \in Z} \exists x . \zeta \wedge 
 \forall x \forall y (x = y \vee \eta) \wedge
       \bigwedge_{h=0}^{m-1} \forall x \exists y (x \neq y \wedge \theta_h), 
\end{split}
\label{eq:nf}
\end{equation}
where $Z$ is a finite set of unary quantifier- and equality-free formulas and the other components are as in~\eqref{eq:snf}. We refer to the parameter $m$ in both~\eqref{eq:snf} and~\eqref{eq:nf} as the {\em multiplicity} of $\phi$.  

The following basic fact about $\FOT$ goes back, essentially, to~\cite{fo21t:scott62}, and is
widely used in studies of $\FOT$ and its variants~\cite[Lemma~8.1.2]{fo21t:BGG}. Remembering our general assumption that all structures have cardinality at least 2, we have:
\begin{lemma}
Let $\phi$ be an $\FOT$-formula. There exists a standard normal-form $\FOT$-formula $\phi'$ such that:
\textup{(i)} $\models \phi' \rightarrow \phi$; \textup{(ii)} every model of $\phi$ can be expanded to a model
of $\phi'$; and \textup{(iii)} $\sizeOf{\phi'}$ is
bounded by a polynomial function of $\sizeOf{\phi}$.
\label{lma:nfL2}
\end{lemma}
Obviously, Lemma~\ref{lma:nfL2} applies without change to $\LtoPO$ and $\LtoPOu$. Under our general restriction to structures with at least 2 elements, $\exists x . \zeta$ is logically equivalent to $\forall x \exists y (x \neq y \wedge  (\zeta \vee \zeta(y)))$. Hence any formula in
weak normal form can be converted, in polynomial time, to a logically equivalent one in standard normal form. However, this process increases the
multiplicity of the formula in question: in the sequel, we shall sometimes need~\eqref{eq:nf} in full generality, in order to obtain 
finer control over this parameter.

\section{Unary two-variable logic with one partial order}
\label{sec:L21pou}
The purpose of this section is to show that the logic $\LtoPOu$ has the doubly 
exponential-sized finite model property (Theorem~\ref{theo:mainPOUnary}): if $\phi$ is a finitely satisfiable $\LtoPOu$-formula, then $\phi$ has a model
of size bounded by some fixed doubly exponential function of $\sizeOf{\phi}$. It follows that the finite satisfiability problem for $\LtoPOu$ is in \TwoNExpTime. 

All structures in this section interpret a signature of unary predicates, together with the
distinguished  predicate $<$. To make reading easier, we typically write $x > y$ for $y < x$ and
$x \sim y$ for $\neg (x = y \vee x < y \vee y < x)$. In practice, we will simply treat the symbols $>$ and $\sim$ as if they were binary predicates (subject to the obvious constraints on their interpretations). With this concession to informality, we see that, in
the logics $\LtoPO$ and $\LtoPOu$, 
any pair of distinct elements of a structure satisfies exactly one of the atomic formulas
$x < y$, $x > y$ or $x \sim y$. Where a structure $\fA$ is clear from context, we typically do not distinguish between the predicate $<$ and its interpretation in $\fA$, writing $a < b$ to mean $\langle a, b \rangle \in <^\fA$; similarly for $>$, $\sim$ and $=$.
We sometimes refer to the distinguished predicates $\ft$, $<$, $>$, $\sim$ and $=$ as {\em navigational} predicates. (The allusion here is to the terminology employed in XPATH.) A predicate that is not navigational is called {\em ordinary}.
A formula is navigation-free if it contains no navigational predicates. 
A formula is said to be \textit{pure Boolean} if it is quantifier-, and navigation-free---i.e.~if it is a Boolean combination of literals featuring ordinary predicates. 
Notice that all 1-types contain the conjuncts $\neg x < x$ and $x = x$, and hence are not, technically speaking,
pure Boolean formulas. However, they are of course logically equivalent to the pure Boolean formulas obtained
by deleting all navigational conjuncts.

In this section, we use the (possibly decorated) variables $\alpha$, $\beta$, $\gamma$, $\pi$ to range over 1-types, $\mu$ over unary pure Boolean formulas, and $\zeta$, $\eta$, $\theta$, $\phi$,  $\chi$, $\psi$ over other $\LtoPOu$-formulas.

\subsection{Basic formulas}
Structures interpreting $\LtoPOu$-formulas have a very simple form, and it will be convenient to diverge 
slightly from  standard model-theoretic terminology when discussing them. (Remember, all structures are taken to have cardinality at least 2 in this paper.)
Let $\Pi$ be a fixed set of 1-types over some unary signature $\sigma$.
A {\em typed partial order} (\textit{over} $\Pi$) is a triple $\fA = (X,<,\tp)$, where $X$ is a set of cardinality at least 2, $<$ a partial order on $X$, and
$\tp: X \rightarrow \Pi$ a function. We can regard $\fA$ as a structure 
interpreting $\LtoPOu$-formulas in the obvious way; and it is evident that all
structures interpreting $\LtoPOu$-formulas can be regarded as typed partial orders over some set of 1-types. This is what we shall do in the sequel, therefore. If $\fA = (X,<,\tp)$ is a typed partial order and $a \in X$, we call $a$ {\em maximal} if it a largest element of its 1-type, i.e.~there exists no $a'$ such that $\tp(a') = \tp(a)$ and $a < a'$; similarly, {\em mutatis mutandis}, for {\em minimal}. We call $a$ {\em extremal} if it is either maximal or minimal.

We begin by establishing a stronger normal form theorem for $\LtoPOu$. 
We call a formula {\em basic} if it has one of the forms
\begin{align}
& \forall x(\alpha \rightarrow \forall y (\alpha(y) \rightarrow x=y))
\tag{B1a}
\label{eq:b1a}
\\
& \forall x(\alpha \rightarrow \forall y (\beta(y) \rightarrow x=y))
\tag{B1b}
\label{eq:b1b}
\\
& \forall x(\alpha \rightarrow \forall y (\alpha(y) \wedge x\neq y \rightarrow  x \sim y))
\tag{B2a}
\label{eq:b2a}
\\
& \forall x(\alpha \rightarrow \forall y (\beta(y) \rightarrow  x \sim y))
\tag{B2b}
\label{eq:b2b}
\\
&\forall x(\alpha \rightarrow \forall y (\beta(y) \rightarrow  x < y))
\tag{B3}
\label{eq:b3}
\\
& \forall x(\alpha \rightarrow \forall y (\beta(y) \rightarrow  (x < y \vee x \sim y))
\tag{B4}
\label{eq:b4}
\\
& \forall x(\alpha \rightarrow \forall y (\alpha(y) \wedge x \neq y \rightarrow  (x < y \vee x > y))
\tag{B5a}
\label{eq:b5a}
\\
& \forall x(\alpha \rightarrow \forall y (\beta(y) \rightarrow  (x < y \vee x > y))
\tag{B5b}
\label{eq:b5b}
\\
& \forall x(\alpha \rightarrow \exists y (\mu(y) \wedge \neg \alpha(y) \wedge x < y))
\tag{B6}
\label{eq:b6}
\\
& \forall x(\alpha \rightarrow \exists y (\mu(y) \wedge \neg \alpha(y)\wedge x > y))
\tag{B7}
\label{eq:b7}
\\
& \forall x(\alpha \rightarrow \exists y (\mu(y)  \wedge x \sim y))
\tag{B8}
\label{eq:b8}
\\
& \forall x . \mu
\tag{B9}
\label{eq:b9}
\\
& \exists x . \mu,
\tag{B10}
\label{eq:b10}
\end{align}
where $\alpha$ and $\beta$ are distinct 1-types and $\mu$ is a unary pure Boolean formula.
We typically use the variable $\psi$ to range over basic formulas and $\Psi$ to range over finite sets of basic formulas. Formulas of the forms~\eqref{eq:b3} and~\eqref{eq:b5b} receive special treatment in the sequel, and will be referred to---for reasons that will become evident---as {\em factor-controllable} formulas. If $\Psi$ is any finite set of basic formulas, we denote by $\FC(\Psi)$ the set of factor-controllable formulas in $\Psi$.

\begin{lemma}
Let $\phi$ be a weak normal-form $\LtoPOu$-formula with multiplicity $m$ over a signature $\sigma$.
There exists an $\LtoPOu$-sentence $\phi^*$ over a signature $\sigma^*$, such that: \textup{(}i\textup{)} $\phi$ and $\phi^*$ are satisfiable over the same finite domains;
\textup{(}ii\textup{)} $|\sigma^*| = |\sigma|+3m$; and
\textup{(}iii\textup{)} 
$\phi^*$ is a conjunction of basic formulas.
\label{lma:nfPO}
\end{lemma}
\begin{proof}
Let $\phi$ be as given in~\eqref{eq:nf}. The conjuncts $\exists x. \zeta$ are already of the form~\eqref{eq:b10}, and so
require no action.
Consider next any conjunct $\chi_h = \forall x \exists y (x \neq  y \wedge \theta_h)$, where $0 \leq h < m$. Letting 
let $p_{h,<}$, $p_{h,>}$ and $p_{h,\sim}$ be fresh unary predicates, we may replace $\chi_h$ by the conjunction $\chi^*_h$ of the formulas
\begin{align}
& \forall x (p(x) \rightarrow p_{h,<}(x) \vee p_{h,>}(x) \vee p_{h,\sim}(x))
%\tag{A3.1}
\label{eq:a2split}\\
& \forall x(p_{h,<}(x) \rightarrow \exists y (\theta_h \wedge x < y))
%\tag{A3.2}
\label{eq:a2less}\\
& \forall x(p_{h,>}(x) \rightarrow \exists y (\theta_h \wedge x > y))
%\tag{A3.3}
\label{eq:a2greater}\\
& \forall x(p^+_{h,\sim}(x) \rightarrow \exists y (\theta_h \wedge x \sim y)).
%\tag{A3.4}
\label{eq:a2incomparable}
\end{align}
Obviously, $\models \chi^*_h \rightarrow \chi_h$; moreover, 
any model $\fA$ of $\chi_h$ can be expanded to a model $\fA'$ of $\chi_h$ by setting
$p_{h,<}^{\fA'} = 
\set{a \in A \mid \text{there exists $b < a$ s.t.~$\fA \models \theta_h[a,b]$}}$,
and similarly for $p_{h,>}$ and $ p_{h,\sim}$.
Carrying out this replacement for all $h$ ($0 \leq h < m$), let $\sigma^*$ denote the enlarged signature. Evidently, $|\sigma^*| = |\sigma| + 3m$. 

Formula~\eqref{eq:a2split} is of the form~\eqref{eq:b9}. 
Now replace any formula of the form~\eqref{eq:a2less}
by the conjunction of all formulas of the forms
$\forall x(\alpha(x) \rightarrow \exists y ([\theta_h/\alpha] \wedge x < y))$, where $\alpha$ ranges over the  set of 1-types (over $\sigma^*$) containing $p_{<,h}(x)$, and $[\theta_h/\alpha]$ denotes the result of
replacing each ordinary literal $q(x)$ in $\theta_h$ by $\top$ or $\bot$ as determined by $\alpha(x)$. 
Doing the same for 
\eqref{eq:a2greater} and~\eqref{eq:a2incomparable}, and replacing any navigational literals in $[\mu_h/\alpha]$
by $\top$ or $\bot$ in the obvious way yields logically equivalent conjunctions 
of formulas of the respective forms
\begin{align}
& \forall x(\alpha \rightarrow \exists y (\mu(y) \wedge x < y))
\tag{B6$'$}
\label{eq:b6prime}
\\
& \forall x(\alpha \rightarrow \exists y (\mu(y) \wedge x > y))
\tag{B7$'$}
\label{eq:b7prime}
\\
& \forall x(\alpha \rightarrow \exists y (\mu(y) \wedge x \sim y)),
\tag{B8$'$}
\label{eq:b8prime}
\end{align}
where $\mu$ is a quantifier and navigation-free formula not involving the variable $x$. Notice however that, over finite structures $\fA$, \eqref{eq:b6prime} entails $\forall x(\alpha \rightarrow \exists y (\mu(y) \wedge \neg \alpha(y) \wedge x < y))$. This is obvious since, if $\fA \models \alpha[a]$, let $a'$ be
a maximal element of 1-type $\alpha$ above $a$. That is: $a \leq a'$, $\fA \models \alpha[a']$ and
there does not exist $a''$ such that $a' < a''$ and $\fA \models \alpha[a'']$. By~\eqref{eq:b6prime}, 
let $b$ be such that $a' < b$ and $\fA \models \mu[b]$. But then $\fA \models \neg \alpha[b]$ and $a < b$, 
as required. Thus, \eqref{eq:b6prime} can be replaced by~\eqref{eq:b6}. Likewise,
\eqref{eq:b7prime} can be replaced by~\eqref{eq:b7}. Notice that~\eqref{eq:b8prime} is just~\eqref{eq:b8}.

Consider finally the conjunct $\chi = \forall x \forall y (x = y \vee \eta)$ of $\phi$.
Clearly, we may replace this formula by the conjunction $\chi^*$ of all formulas of the forms
\begin{align*}
& \forall x(\alpha \rightarrow \forall y (\beta(y) \wedge x \neq y \rightarrow [\eta/(\alpha,\beta)]),
\end{align*}
where $\alpha$ and $\beta$ range over the 
set of 1-types (over $\sigma^*$), and $[\eta/(\alpha,\beta)]$ denotes the result of
replacing each unary literal in $\eta$ by its truth-value as determined by $\alpha$ and $\beta(y)$. Clearly, $\models \chi \leftrightarrow \chi^*$. Furthermore, any sub-formula $[\eta/(\alpha,\beta)]$ features only the navigational predicates
$>$, $<$ and $\sim$,
and thus is logically equivalent to one of the forms
$\bot$, $x \sim y$, $x > y$, $x < y$, $(x > y \vee x \sim y)$, $(x <y \vee x \sim y)$,
$(x < y \vee x > y)$ or $\top$. 
Ignoring the trivial case $\top$, and exchanging the variables $x$ and $y$ if necessary, we obtain the forms 
\begin{align}
& \forall x(\alpha \rightarrow \forall y (\beta(y) \wedge x \neq y \rightarrow \bot))
\tag{B1$'$}
\label{eq:4.1}\\
& \forall x(\alpha \rightarrow \forall y (\beta(y) \wedge x \neq y \rightarrow x \sim y))
\tag{B2$'$}
\label{eq:4.2}\\
& \forall x(\alpha \rightarrow \forall y (\beta(y) \wedge x \neq y \rightarrow x < y))
\tag{B3$'$}
\label{eq:4.3}\\
& \forall x(\alpha \rightarrow \forall y (\beta(y) \wedge x \neq y \rightarrow (x <y \vee x \sim y)).
\tag{B4$'$}
\label{eq:4.4}\\
& \forall x(\alpha \rightarrow \forall y (\beta(y) \wedge x \neq y \rightarrow (x <y \vee x > y)).
\tag{B5$'$}
\label{eq:4.5}
\end{align}

We consider these forms in turn, according as
$\alpha$ and $\beta$ are identical or distinct. 
For~\eqref{eq:4.1}, we have~\eqref{eq:b1a} and ~\eqref{eq:b1b}.
For~\eqref{eq:4.2}, we obtain~\eqref{eq:b2a}  and~\eqref{eq:b2b}. 
For~\eqref{eq:4.3}, if $\alpha = \beta$, we have~\eqref{eq:b1a} again; if $\alpha \neq \beta$, we have~\eqref{eq:b3}. 
For~\eqref{eq:4.4}, if $\alpha = \beta$, we have~\eqref{eq:b2a} again; if $\alpha \neq \beta$, we have~\eqref{eq:b4}. 
For~\eqref{eq:4.5}, we obtain~\eqref{eq:b5a}  and~\eqref{eq:b5b}.
\end{proof}

\subsection{Factorizations}
\label{subsec:factorizations}

The following notion will play a crucial role in the sequel. Let 
$\fA =$ $(X,<,\tp)$ be a typed partial order. A {\em factorization} of $\fA$ is
a pair $\bbB = (\bB, \ll)$, where $\bB$ is a partition of $X$, and 
$\ll$ is a partial order on $\bB$ satisfying:
\begin{enumerate}[(F1)]
\item for all $B \in \bB$, there exists $\pi \in \Pi$, denoted
      $\tp(B)$, such that,
	  for all $b \in B$, $\tp(b) = \pi$; \label{enum:F1}
\item for all $\pi \in \Pi$, the set $\{B \in \bB \mid \tp(B) = \pi\}$ is linearly ordered by $\ll$;
\item for all $A, B \in \bB$, if $A \ll B$, then, for all $a \in A$ and all $b \in B$, $a < b$.
\end{enumerate}
We refer to the elements of $\bB$ as {\em blocks}, and to the ordering $\ll$ as the {\em block ordering} (in contradistinction to the {\em element ordering} $<$). Notice that, if $|\bB| \geq 2$, the triple $(\bB, \ll, \tp)$ is itself a typed partial order.
If $\tp(B) =\alpha$, we call $B$ an $\alpha$-block.
We say that a block $B$ is {\em of type $\alpha \vee \beta$} if it is either of type $\alpha$ or of type $\beta$, and
we call $B$ an $(\alpha \vee \beta)$-block. 

In the context of a factorization $(\bB, \ll)$, we use $A \gg B$ as an alternative to $B \ll A$. 
If $A$ and $B$ are blocks, we write $A \approx
B$ to mean that $A$ and $B$ are distinct and neither $A \ll B$ nor $B \ll A$.
Thus, $\approx$ stands in the same relation to $\ll$ as $\sim$ does to $<$.
Note that, if $A \approx B$, it is possible for there to be $a, a' \in A$ and $b, b' \in B$ such that
$a < b$ and $a' > b'$.
A block $B$ is \textit{maximal} if there exists no block $B'$ such that $\tp(B) = \tp(B')$ and $B \ll B'$; similarly for {\em minimal}. A block
is {\em extremal} if it is either maximal or minimal. 
We denote the set of extremal blocks of $\bB$ by $\bB^\times$.

Thus, a factorization of a typed partial order is an organization of its elements into blocks of uniform type, with a partial order on the blocks such that all blocks of a given type are linearly ordered, and, such that, whenever one block is less than another in the block ordering, every element of the first block is less than every element of the second in the element ordering. 
Fig.~\ref{fig:factorization} shows a factorization of a finite typed partial order over 1-types $\pi_1, \dots, \pi_N$, depicted
as an acyclic directed graph: the block order $\ll$ is the transitive closure of the edges; extremal
blocks are marked with thick boundaries. The shaded blocks and the line marked $\chi$ will be explained in
Sec.~\ref{subsec:blockCut}. 
	
It is important to realize that the factorization 
$(\bB, \ll)$ does not determine the partial order $(X,<)$. Indeed, any typed partial order $\fA$ has a factorization, namely, the trivial factorization 
in which the blocks are simply the non-empty sets $\set{a \in X \mid \tp(a) = \pi}$ for $\pi \in \Pi$, and the block-order is empty. The next two lemmas show that we can generally find more informative factorizations than this. Recall in this context that a  factor-controllable basic formula is one of either of the forms~\eqref{eq:b3} or~\eqref{eq:b5b}.

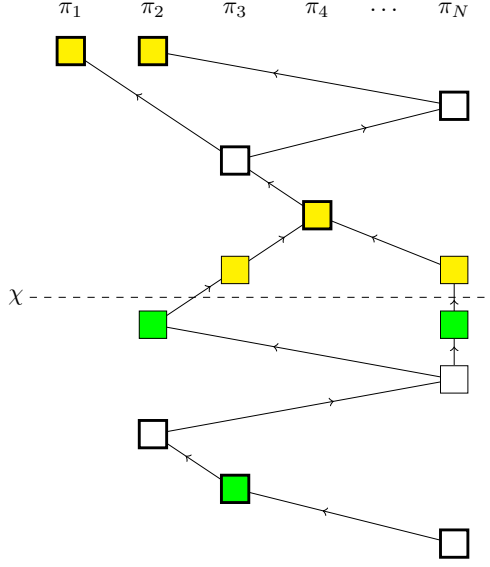
\begin{figure}
\begin{center}
\resizebox{!}{7.5cm}{\begin{tikzpicture}[scale=0.4]

\coordinate [label=center:{$\pi_1$}] (pi1) at (0.5,20);
\coordinate [label=center:{$\pi_2$}] (pi2) at (3.5,20);
\coordinate [label=center:{$\pi_3$}] (pi3) at (6.5,20);
\coordinate [label=center:{$\pi_4$}] (pi4) at (9.5,20);
\coordinate [label=center:{$\dots$}] (dots) at (12.0,20);
\coordinate [label=center:{$\pi_N$}] (piN) at (14.5,20);

\draw[->-=0.6] (14.5,16.5) -- (3.5,18.5);
\draw[->-=0.6] (6.5,14.5) -- (0.5,18.5);
% \draw[->-=0.6] (6.5,14.5) -- (3.5,18.5);
\draw[->-=0.6] (6.5,14.5) -- (14.5,16.5);
\draw[->-=0.6] (9.5,12.5) -- (6.5,14.5);
\draw[->-=0.6] (6.5,10.5) -- (9.5,12.5);
\draw[->-=0.6] (14.5,10.5) -- (9.5,12.5);
\draw[->-=0.7] (3.5,8.5) -- (6.5,10.5);
\draw[->-=0.45] (14.5,8.5) -- (14.5,10.5);
\draw[->-=0.6] (14.5,6.5) -- (3.5,8.5);
\draw[->-=0.6] (14.5,6.5) -- (14.5,8.5);
\draw[->-=0.6] (3.5,4.5) -- (14.5,6.5);
\draw[->-=0.6] (6.5,2.5) -- (3.5,4.5);
\draw[->-=0.6] (14.5,0.5) -- (6.5,2.5);

\filldraw[very thick,fill=yellow] (0,18) rectangle (1,19);
\filldraw[very thick,fill=yellow] (3,18) rectangle (4,19);
\filldraw[very thick,fill=white] (14,16) rectangle (15,17);
\filldraw[very thick,fill=white] (6,14) rectangle (7,15);
\filldraw[very thick,fill=yellow] (9,12) rectangle (10,13);
\filldraw[fill=yellow] (6,10) rectangle (7,11);
\filldraw[fill=yellow] (14,10) rectangle (15,11);
\filldraw[fill=green] (3,8) rectangle (4,9);
\filldraw[fill=green] (14,8) rectangle (15,9);
\filldraw[fill=white] (14,6) rectangle (15,7);
\filldraw[very thick,fill=white] (3,4) rectangle (4,5);
\filldraw[very thick,fill=green] (6,2) rectangle (7,3);
\filldraw[very thick,fill=white] (14,0) rectangle (15,1);

% Cut

\draw[dashed] (-1,9.5) -- (16,9.5);
\coordinate [label=center:{$\chi$}] (chi) at (-1.5,9.5);

\end{tikzpicture}}
\end{center}
\caption{Factorization of a finite typed partial order over 1-types $\pi_1, \dots, \pi_N$, with cut $\chi$, showing
$\bB^\times$ (thick lines), and $F^+(\chi)$ and $F^-(\chi)$ (shading).
}
\label{fig:factorization}
\end{figure}

%Factorizations help us concentrate on those aspects of typed partial orders which are responsible for the satisfaction of certain formulas. The following lemma provides a trivial illustration.

\begin{lemma}
Let $\fA$ be a typed partial order over a set of types $\Pi$, with $\alpha, \beta \in \Pi$ distinct.
Suppose $\bbB$ is a factorization of $\fA$ in which every block of type $\alpha$ lies below every block of type $\beta$ in the block order. Then the formula~\eqref{eq:b3}, namely
\begin{equation*}
\forall x(\alpha(x) \rightarrow \forall y (\beta(y) \rightarrow  x < y)),
\end{equation*}
is true in $\fA$. Conversely, if~\eqref{eq:b3} is true in $\fA$, then there exists a factorization $\bbB$ of $\fA$  in which every block of type $\alpha$ lies below every block of type $\beta$ in the block order.
\label{lma:lmaBlockAlphaBetaGlobal}
\end{lemma}
\begin{proof}
The first statement of the lemma is obvious. For the converse,
let $\bB$ consist of the non-empty sets $\set{a \in X \mid \tp(a) = \pi}$ for $\pi \in \Pi$. If
there is no $\alpha$-block or no $\beta$-block, let $\ll$ be the empty partial order. Otherwise,
 let $A \in \bB$ be the $\alpha$-block, let $B \in \bB$ be
the $\beta$-block, and let
$\ll = \set{\langle A, B \rangle}$. 
\end{proof}
If the typed partial order $\fA$ is clear from context, and $\bbB$ is a factorization of $\fA$, 
we write $\bbB \models \forall x(\alpha(x) \rightarrow \forall y (\beta(y) \rightarrow  x < y))$ to mean that every block of
type  $\alpha$ is less than every block of
type $\beta$ in the block order. The motivation for this
notation should be obvious from Lemma~\ref{lma:lmaBlockAlphaBetaGlobal}.
	
\begin{lemma}
Let $\fA$ be a typed partial order over a set of types $\Pi$, with $\alpha, \beta \in \Pi$ distinct.
Suppose $\bbB$ is a factorization of $\fA$ in which the set of blocks of type $\alpha \vee \beta$ is linearly ordered. 
Then the formula~\eqref{eq:b5b}, namely
\begin{equation*}
\forall x(\alpha(x) \rightarrow \forall y (\beta(y) \rightarrow  (x < y \vee y < x))),
\end{equation*}
is true in $\fA$. 
Conversely, if~\eqref{eq:b5b} is true in $\fA$, then there exists a factorization $\bbB$ of $\fA$ such that the set of blocks of type 
$\alpha \vee \beta$ is linearly ordered.
\label{lma:lmaBlockAlphaBetaAlternate}
\end{lemma}
\begin{proof}
The first statement of the lemma is obvious. 
For the converse, let $A_0 = \set{a \in X \mid \tp(a) = \alpha}$ and $B_0 = \set{a \in X \mid \tp(a) = \beta}$.
We may assume that both these sets are non-empty, since otherwise the trivial factorization satisfies the conditions of the lemma. 
Define an equivalence relation $\equiv$ on $A_0$ by setting $a \equiv a'$ if, for all $b \in B_0$, $a < b \Leftrightarrow a' < b$. Similarly, define an equivalence relation $\equiv$ on $B_0$ by setting $b \equiv b'$ if, for all 
$a \in A_0$, $b' <a  \Leftrightarrow b' < a$. 
Let $\bB$ be the partition of $X$ whose cells are: (i)
the equivalence classes of $\equiv$ in $A_0$, (ii) the equivalence classes of $\equiv$ in $B_0$, and (iii) the non-empty sets $\set{a \in X \mid \tp(a) = \pi}$, where 
$\pi \in \Pi \setminus \set{\alpha, \beta}$. For any $C, D \in \bB$ write $C \ll D$ just in case 
$C$ and $D$ are distinct $(\alpha \vee \beta)$-blocks such that
there exist $c \in C$ and $d \in D$ with $c < d$. 

To show that $\bbB = (\bB, \ll)$ has the desired properties, we first observe that, if $A$ is an $\alpha$-block and $B$ a $\beta$-block, then $A \ll B$ if and only if, for all $a \in A$ and $b \in B$,
$a < b$. Now suppose that $A$ and $A'$ are distinct $\alpha$-blocks. Let $b_0$ be some element 
such that there exist $a_0 \in A$ and $a_1 \in A'$ such that either $a_0 < b_0 < a_1$ or
$a_1 < b_0 < a_0$. In the former case, $a < a'$ for all $a \in A$ and $a' \in A'$; and in the latter,
$a > a'$ for all $a \in A$ and $a' \in A'$. Similar remarks apply to $\beta$-blocks. Thus,
for any 
$(\alpha \vee \beta)$-blocks $C$ and $D$, $C \ll D$ if and only if, for all $c \in C$ and $d \in D$,
$c < d$. It follows that $\ll$ is a partial order in which the collection of $(\alpha \vee \beta)$-blocks is linearly ordered. Therefore, the collection of 
$\alpha$-blocks and the collection of  $\beta$-blocks are also both linearly ordered. For $\gamma \in  \Pi \setminus\set{\alpha, \beta}$, there is at most one $\gamma$-block. That is, $(\bB,\ll)$ is a factorization of $\fA$.

\end{proof}
If the typed partial order $\fA$ is clear from context, and $\bbB$ is a factorization of $\fA$,  we write $\bbB \models \forall x(\alpha(x) \rightarrow \forall y (\beta(y) \rightarrow  (x < y \vee x > y))$ to mean that
the set of blocks of type  $\alpha \vee \beta$ is linearly ordered. The motivation for this
notation should be obvious from Lemma~\ref{lma:lmaBlockAlphaBetaAlternate}.

Suppose $\fA$ is a typed partial order and $\bbB_1 = (\bB_1, \ll_1)$, $\bbB_2= (\bB_2, \ll_2)$ are factorizations of $\fA$. We say that  $\bbB_2$ is a \textit{refinement} of $\bbB_1$ if, for all $A_2 \in \bB_2$, there exists a (necessarily unique) $A_1 \in \bB_1$ such that $A_2 \subseteq A_1$, and moreover, for all $A_2, B_2 \in \bB_2$,
and all $A_1, B_1 \in \bB_1$ such that $A_2 \subseteq A_1$ and $B_2 \subseteq B_1$, $A_1 \ll_1 B_1$ implies $A_2 \ll_2 B_2$.    
\begin{lemma}
Any two factorizations of a typed partial order have a common refinement.
\label{lma:commonRefinement}
\end{lemma}
\begin{proof}
Let $\bbB_1 = (\bB_1, \ll_1)$ and  $\bbB_2 = (\bB_2, \ll_2)$ be factorizations of the typed partial order $\fA$. Define  $\bbB = (\bB, \ll)$ as follows:
let 
\begin{equation*}
\bB = \set{B_1 \cap B_2 \mid B_1 \in \bB_1, B_2 \in \bB_2} \setminus \set{\emptyset}; 
\end{equation*}
and 
let $\ll$ be the transitive closure of the relation 
\begin{equation*}
\set{\langle A, B  \rangle \mid A= A_1 \cap A_2, B = B_1 \cap B_2,
	A_1 \ll_1 B_1 \mbox{ or } A_2 \ll_2 B_2}.
\end{equation*}
A simple induction shows that, if $A$ is related to $B$ by $\ll$ then, for all $a \in A$ and all $b \in B$, 
$a < b$. It follows that $\ll$ is irreflexive and hence is a partial order. It is thus immediate
from the definition of $\ll$ that the partially ordered set $\bbB = (\bB, \ll)$ is a factorization of $\fA$ and moreover that it is a refinement of both $\bbB_1$ and $\bbB_2$.
\end{proof}
Refinements of block orders are useful because they preserve the properties featured in Lemmas~\ref{lma:lmaBlockAlphaBetaGlobal} and~\ref{lma:lmaBlockAlphaBetaAlternate}. The following Lemma is immediate.
\begin{lemma}
Let $\fA$ be a typed partial order and $\bbB$, $\bbB'$ factorizations of $\fA$ with $\bbB'$ a refinement of\/ $\bbB$. Let $\psi$ be a factor-controllable basic formula. If $\bbB \models \psi$, then $\bbB' \models \psi$.
\label{lma:refinePreserve}
\end{lemma} 

A {\em unit block} of $\bbB$ is a block containing exactly one element of $A$. Trivially, every unit block
is linearly ordered by $<$. We say that $\bbB$ is \textit{unitary} if every block of $\bB$ which is linearly ordered by $<$ is a unit block.
Combining all of the above lemmas, we have:
\begin{lemma}
Let $\fA$ be a typed partial order and 
$\Psi$ a finite set of basic formulas such that $\fA \models \Psi$. Then there is  a unitary factorization $\bbB$ of $\fA$ such that $\bbB \models \FC(\Psi)$.
\label{lma:factorize}
\end{lemma}
\begin{proof}
For each $\psi \in \FC(\Psi)$, we apply Lemmas~\ref{lma:lmaBlockAlphaBetaGlobal} or~\ref{lma:lmaBlockAlphaBetaAlternate} as appropriate, and take a common refinement of all the resulting factorizations by Lemma~\ref{lma:commonRefinement}. Now further refine by replacing all linearly ordered blocks with unit blocks having the
obvious block order. The result then follows by Lemma~\ref{lma:refinePreserve}.
\end{proof}

Let $\fA = (X,<, \tp)$ be a typed partial order and $\bbB = (\bB, \ll)$ a factorization of $\fA$. 
We have already observed that $\bbB$ does not contain all the information required to reconstruct the element order $<$. However, it very nearly does, in a sense that we can make precise.
Let us first overload the block-order $\ll$ by writing for all
$a, b \in X$, $a \ll b$ if there exist $A, B \in \bB$ such that $a \in A$, $b \in B$ and $A \ll B$. We might call $\ll$ the \textit{inter-block order} on $X$. It is obvious that the inter-block order is a partial order, and, from {(B3)}, that it is contained in the element order $<$. Now, for all $a, b \in X$, write $a <_0 b$ if $a <b$, and both $a$ and $b$ belong to the same block of $\bB$. Again, this is clearly a partial order: we call it the {\em intra-block order}. Finally, define $a <_\times b$ if $a <b$, and both $a$ and $b$ are both extremal elements of $(X,<)$;  
once again, $<_\times$ is clearly a partial order: we call it the {\em extremal order}. Now define the binary relation
$\lessdot$ on $X$ to be the transitive closure of
$(\ll \cup <_0 \cup <_\times)$. It is obvious that $\lessdot$ is a partial order no stronger than (i.e.~included in) $<$, but that, nevertheless, $\bbB$ is a factorization of the typed partial order
$(X,\lessdot, \tp)$. It is also obvious that, when restricted to elements of some fixed 1-type $\pi$,
$<$ and $\lessdot$ coincide.
We say that $\fA$ is {\em thin}  over $\bbB$ if $<$ and $\lessdot$ coincide over the whole of $X$. 
\begin{lemma}
	Suppose $\fA = (X,<, \tp)$ is a finite typed partial order and $\Psi$ a set of basic formulas such that $\fA \models \Psi$. Let 
	$\bbB$ be a factorization of $\fA$  such that $\bbB \models \FC(\Psi)$. Then there exists a typed partial order $\dot{\fA}$ over the domain $X$ such that $\bbB$ is a factorization of $\dot{\fA}$, 
	$\dot{\fA} \models \Psi$, and $\dot{\fA}$ is thin over $\bbB$.  
	\label{lma:loseFiveKG}
\end{lemma}
\begin{proof}
	Define $\lessdot$ to be the transitive closure of
	$(\ll \cup <_0 \cup <_\times)$, as just described, and let
	$\dot{\fA} = (X,\lessdot, \tp)$. Thus, $\bbB$ is a factorization of $\dot{\fA}$, with $\dot{\fA}$ thin over $\bbB$.
	We show that $\dot{\fA} \models \psi$, where $\psi \in \Psi$ is of each of the possible forms~\eqref{eq:b1a}--\eqref{eq:b10} in turn. 
		\begin{description}
		\item \eqref{eq:b1a}, \eqref{eq:b1b}, \eqref{eq:b9}, \eqref{eq:b10}: $\psi$ does not involve the ordering.
		\item \eqref{eq:b2a}, \eqref{eq:b2b}, \eqref{eq:b4}, \eqref{eq:b8}: $\lessdot$ is no stronger than $<$ \ .
		\item \eqref{eq:b3}, \eqref{eq:b5b}: $\bbB \models \psi$.
		\item \eqref{eq:b5a}: When restricted to elements of some fixed type,		
		$<$ and $\lessdot$ coincide.
		\item \eqref{eq:b6}: Suppose that $a \in X$ is of type $\alpha$. Since $<$ and $\lessdot$ coincide on elements of some fixed type, let $a^*$ be a maximal
		element of type $\alpha$ such that either $a = a^*$ or $a < a^*$ (equivalently: $a = a^*$ or $a \lessdot a^*$). But $\fA \models \psi$, so there exists an element $b$ satisfying $\mu$---say of type $\beta \neq \alpha$---such that $a^* < b$, and hence
		a maximal element $b^*$ of type $\beta$ such that $a^* <_\times b^*$. But then 
		$a \lessdot b^*$, whence $\dot{\fA} \models \psi$.
		\item \eqref{eq:b7}: Similar to~\eqref{eq:b6}.	
	\end{description}
\end{proof}

\subsection{Reducing the number of blocks}
\label{subsec:blockCut}
Suppose $\fA = (X,<,\tp)$ is a finite typed partial order with factorization
\mbox{$\bbB = (\bB, \ll)$}. The following notions will help us to reason about $\bbB$. 
Recall that $\bB^\times$ denotes the set of extremal blocks of $\bB$.
If $B \in \bB$, define the {\em depth} of $B$, denoted $d(B)$, to be the length $m$ of the longest path $B= B_0 \ll \cdots \ll B_m$.   
The {\em depth} of $\bbB$, denoted $d(\bbB)$, is the maximum value attained by $d(B)$ for $B \in \bB$. A {\em cut} is a number $\chi = i + 0.5$ where $0 \leq i < d(\bbB)$. If $\chi$ and $\chi'$ are cuts, we say $\chi'$ is \textit{above} $\chi$
(and $\chi$ is \textit{below} $\chi'$)
 if $\chi' < \chi$. (Depth increases as we go down.)  Similarly, if $B \in \bB$, we say that $B$ is {\em above}  $\chi$ if  $d(B) < \chi$, and {\em below} $\chi$ if $d(B) > \chi$. If $\chi'$ is also a cut of $\bbB$ with $\chi$ below $\chi'$, we say that 
$B$ is {\em between} $\chi$ and $\chi'$ if it is above $\chi$ and below $\chi'$. 

For any cut $\chi$, and any 1-type $\pi$, a  
a \textit{minimal} $\pi$-block \textit{above} $\chi$ is a block $B$ 
such that $\tp(B) = \pi$, $d(B) < \chi$ and, for all $B' \in \bB$ such that
$\tp(B') = \pi$ and $d(B') < \chi$, $d(B) \geq d(B')$. A \textit{minimal} block \textit{above} $\chi$ is a minimal $\pi$-block above $\chi$ for some $\pi$. 
The notion of {\em maximal} ($\pi$)-block {\em below} $\chi$ is defined analogously.
Denote by $F^+(\chi)$ the set of {\em minimal} blocks 
above $\chi$, and by $F^-(\chi)$ the set of {\em maximal} blocks 
below $\chi$. Note that $F^+(\chi)$ contains at most one block of each type, and similarly for $F^-(\chi)$. Let $F(\chi)= F^-(\chi) \cup F^+(\chi) \cup \bB^\times$; we call $F(\chi)$ the {\em frontier} of $\chi$. 
Fig.~\ref{fig:factorization} shows a cut $\chi = 4.5$ in a factorization of a typed partial order over 
$\pi_1, \dots, \pi_N$. The sets of blocks $F^+(\chi)$ and $F^-(\chi)$ are shown by shading.

If $\chi$ and $\chi'$ are cuts of $\bbB$, with $\chi$ below $\chi'$, we say that $\chi$ and $\chi'$ are \textit{equivalent} if 
there exists a function $f: F(\chi) \rightarrow F(\chi')$ satisfying the following conditions;
\begin{description}
\item (E1): $f$ maps $F^-(\chi)$ to $F^-(\chi')$, $f$ maps $F^+(\chi)$ to $F^+(\chi')$, and $f$ is the identity on $\bB^\times$;
\item (E2): $f: F(\chi) \rightarrow F(\chi')$ is a typed partial order isomorphism, i.e.,~$f$ is 1--1 and onto,
for all $B \in F(\chi)$, $\tp(B) = \tp(f(B))$, and
for all $A, B \in F(\chi)$, and $A \ll B \Leftrightarrow f(A) \ll f(B)$.
\end{description}
Suppose $\chi$ and $\chi'$ are equivalent cuts, with $\chi'$ above $\chi$.
Obviously, no extremal block can lie between $\chi$ and $\chi'$. For example, if 
$A$ is minimal, then $f(A) = A$, $A \in F^+(\chi)$ but $A \not \in F^+(\chi')$, contradicting the requirements of (E1); a similar argument applies if $A$ is maximal.
Equally obviously,
since $F^+(\chi')$ and $F^-(\chi')$ each contain at most one block of any given type,
the function $f$, if it exists, is unique by the requirements of (E2); we denote it by $f_{\chi,\chi'}$.
Observe finally that, for any block $B$ in $F^-(\chi) \cup F^+(\chi)$, the blocks $B$ and $f_{\chi,\chi'}(B)$ stand in the same relations ($\ll$, $\gg$ or $=$) to all extremal blocks.

Fixing $\fA = (X, <, \tp)$ and $\bbB = (\bB,\ll)$, suppose $\chi$, $\chi'$ are equivalent cuts in $\bbB$ with $\chi$ below $\chi'$. 
Let $\bB^- = \set{B \in \bB \mid \text{$d(B) > \chi$}}$ be the set of blocks below $\chi$ and 
$\bB^+ = \set{B \in \bB \mid \text{$d(B) < \chi'$}}$ the set of blocks above $\chi'$. Define $\bB^* = \bB^- \cup \bB^+$, and
define the relation $\ll^*$ on $\bB^*$ to be the transitive closure of the three relations
\begin{align*}
& \set{\langle A,B \rangle \in (\bB^-)^2 \mid A \ll B}\\
& \set{\langle A,B \rangle \in (\bB^+)^2 \mid A \ll B}\\ 
& \set{\langle B,f_{\chi,\chi'}(C) \rangle \mid B \in F^-(\chi), C \in F^+(\chi), B \ll C}.
\end{align*}
Denote by $\bbB^*$ the pair $(\bB^*,\ll^*)$.
Let $X^* = \bigcup \bB^*$, and let $\tp^*$ be the restriction of the function $\tp$ to $X^*$. Noting that $\bB^\times \subseteq \bB^*$, we see that the extremal order $<_\times$ is defined on $X^*$. Let $<^*_0$ be the restriction of the
intra-block order $<_0$ to $X^*$. As before, we overload the symbol $\ll^*$ so that it denotes the inter-block order on $X^*$ under $\bbB^*$ : $a \ll^* b$ if the blocks
$A, B \in \bB^*$ such that $a \in A$, $b \in B$ satisfy $A \ll^* B$. (Note that $\ll^*$ is not in general equal to the restriction to $X^*$ of the inter-block order on $X$ under $\bbB$.)
Finally, let 
$<^*$ be the transitive closure of $(\ll^* \cup <^*_0 \cup <_\times)$. We denote by $\fA^*$ the triple $(X^*,<^*,\tp^*)$.

For Lemmas~\ref{lma:cutBlockWeakening}--\ref{lma:cutBlockOrderPreserve},
we keep $\fA$, $\bbB$, $\chi$ and $\chi'$ fixed, with $\bbB^*$ and $\fA^*$ as defined above.
\begin{lemma}
For all $A, B \in \bB^*$, $A \ll^* B$ implies $A \ll B$. In addition,
for any 1-type $\alpha$, the $\alpha$-blocks of $\bB^*$ are
linearly ordered by $\ll^*$.
\label{lma:cutBlockWeakening}
\end{lemma}
\begin{proof}
The first assertion is immediate from the fact that $\ll^*$ is the transitive closure of three relations all contained in $\ll$.
For the second assertion, observe first that, if there are no $\alpha$-blocks below $\chi$ or above $\chi'$, the result is immediate; hence we may assume otherwise.
It suffices to show that, if $B \in F^-(\chi)$ is the maximal
$\pi$-block in $\bbB$ below $\chi$ and $A$ the minimal $\pi$-block above $\chi'$, then $B \ll^* A$. Let $C \in F^+(\chi)$ be the minimal $\pi$-block above $\chi$, so that $B \ll C$. But then $B \ll^* f_{\chi,\chi'}(C) = A$, and we are done. 
\end{proof}

\begin{lemma}
For all $a, b \in X^*$, $a <^* b$ implies $a < b$.
If, in addition, $\tp^*(a)= \tp^*(b)$, the converse implication holds.
\label{lma:cutWeakening}
\end{lemma}
\begin{proof}
The first assertion is immediate from the first assertion of\linebreak Lemma~\ref{lma:cutBlockWeakening}:
$<^*$ is the transitive closure of three relations all included in $<$.
For the second assertion, suppose $a < b$, and that $a$ and $b$ belong to the respective blocks $A$ and $B$.
If $A = B$, the result follows from the fact that $<^*$ extends the intra-block order $<^*_0$.
Otherwise, we have $A \ll B$, and hence $B \not \ll A$.
By Lemma~\ref{lma:cutBlockWeakening}, $A \ll^* B$ and so $a <^* b$ by the fact that 
$<^*$ extends the inter-block order $\ll^*$.
\end{proof}	

\begin{lemma}
$\fA^*$ is a typed partial order, and
$\bbB^*$ is a factorization of $\fA^*$. Moreover, $\fA^*$ is thin over $\bbB^*$.
\label{lma:cutFundamental}
\end{lemma}
\begin{proof}
By the first assertions of Lemmas~\ref{lma:cutBlockWeakening} and~\ref{lma:cutWeakening}, both $\ll^*$ and $<^*$ are partial orders.
By construction, for all $B \in \bB^*$, every element $b \in B$ satisfies 
$\tp^*(b) = \tp(b) = \tp(B)$. Thus, we can write $\tp^*(B) = \tp(B)$ to denote
the 1-type of $B$ in $\bbB^*$.
The second assertion of Lemma~\ref{lma:cutBlockWeakening} ensures that
all blocks of any fixed 1-type in $\bbB^*$ are linearly ordered. Finally,
the requirement that $A \ll^* B$ implies $a <^* b$ for all $a \in A$, $b \in B$
is secured by the fact that $<^*$ extends the inter-block ordering on $X^*$. The final
statement of the lemma is immediate from the definition of $<^*$.
\end{proof}

\begin{lemma}
If all blocks of type $\alpha$ are less then all blocks of type $\beta$ in the block ordering $\bbB$, 
then the same is true of the block ordering $\bbB^*$. 	
\label{lma:cutPreserveClean}
\end{lemma}
\begin{proof}
We may assume that there exist $\alpha$-blocks and $\beta$-blocks, for otherwise
the lemma is trivial.
Since $\chi$ and $\chi'$ are equivalent, either all $\alpha$-blocks of $\bbB$ are below $\chi$, and there are both $\alpha$- and $\beta$-blocks below $\chi'$, or
all $\beta$-blocks are above $\chi'$, and there are both $\alpha$- and $\beta$-blocks above $\chi$. The result then follows from the definition of
$\ll^*$. 
\end{proof}
\begin{lemma}
If the blocks of type $(\alpha \vee \beta)$ are linearly ordered in $\bbB$, 
then the blocks of this type are linearly ordered in $\bbB^*$. 	
\label{lma:cutPreserveAlternating}
\end{lemma}
\begin{proof}
It suffices to show that, if $A$ is the maximal
$(\alpha \vee \beta)$-block below $\chi$, and $B'$ the minimal $(\alpha \vee \beta)$ block above $\chi'$,
then $A \ll^* B'$. Let $B$ be the minimal 
$(\alpha \vee \beta)$-block above $\chi$. Since 
$\chi$ and $\chi'$ are equivalent, $B$ and $B'$ must be of the same type, so
$f_{\chi,\chi'}(B)=B'$. But $A \ll B$, whence $A \ll^* B'$, as required. 
\end{proof}

\begin{lemma}
Suppose $a \in X^*$ and $b \in X$ are such that $a \sim b$. Then there exists $b' \in X^*$ such that $a \sim^* b'$.
\label{lma:cutBlockOrderPreserve}
\end{lemma}
\begin{proof}
Let $A$ be the block of $\bbB$ containing $a$ and $B_0$ the block containing $b$. 
If $B_0 \in \bB^*$, the 
result follows immediately from
Lemma~\ref{lma:cutWeakening} by setting $b' = b$. 
So we may suppose $B_0$ lies between $\chi$ and $\chi'$. 

Assume, for definiteness, that $A$ lies below $\chi$. Let $\beta = \tp(b) = \tp(B_0)$, and let $B$ be the minimal $\beta$-block above $\chi$. We remark that $\beta \neq \alpha$: for if $a$ and $b$ are of the same 1-type,
then they must be in the same block of $\bbB$, contradicting the supposition that $B_0$ is not in $\bB^*$.
We claim that $A \approx B$.
For either $B = B_0$ or $B \ll B_0$, and certainly $A \approx B_0$, so that $A \not \ll B$; on the other
hand no block $B$ above $\chi$ satisfies $B \ll A$, which proves the claim. 
Now let $B'$ be the minimal $\beta$-block above $\chi'$, so that $f_{\chi,\chi'}(B) = B'$. 
Since no extremal block can lie between the equivalent cuts $\chi$ and $\chi'$, $B$ is not extremal; hence
$B'$ is not extremal, by (E1) and the fact that $f_{\chi,\chi'}$ is injective.
We claim that $A \approx^* B'$. Certainly, $B' \not \ll^* A$, so it suffices to suppose
$A \ll^* B'$, and derive a contradiction. By the construction of $\ll^*$, 
there exist $C$, $D$, $D' = f_{\chi,\chi'}(D)$ such that $C$ is a maximal block below $\chi$,  $D$ is a minimal block above $\chi$, and such that 
$A \ll C$, $C \ll D$ and $D' \ll B'$. Since $B, D \in  F^+(\chi)$
and $B', D' \in  F^+(\chi')$, it follows from (E3) that $D \ll B$. But then $A \ll B$, which is the desired contradiction (see Fig.~\ref{fig:cutBlockOrderPreserve}).
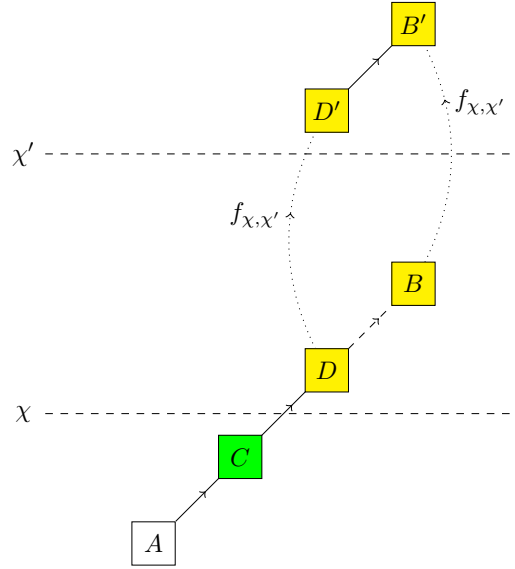
\begin{figure}
\begin{center}
\resizebox{!}{7.5cm}{\begin{tikzpicture}[scale=0.6]

\draw[->-=0.6] (0.5,0.5) -- (2.5,2.5);
\draw[->-=0.6] (2.5,2.5) -- (4.5,4.5);
\draw[->-=0.6,dashed] (4.5,4.5) -- (6.5,6.5);
\draw[->-=0.6] (4.5,10.5) -- (6.5,12.5);
\draw [dotted,->-=0.6, bend left] (4.5,4.5) to node [left,pos=0.6] {$f_{\chi,\chi'}$} (4.5,10.5);
\draw [dotted,->-=0.7, bend right] (6.5,6.5) to node [right,pos=0.7] {$f_{\chi,\chi'}$} (6.5,12.5);

\draw[dashed] (-2,3.5) -- (9,3.5);
\draw[dashed] (-2,9.5) -- (9,9.5);

\filldraw[fill=white] (0,0) rectangle (1,1);
\filldraw[fill=green] (2,2) rectangle (3,3);
\filldraw[fill=yellow] (4,4) rectangle (5,5);
\filldraw[fill=yellow] (6,6) rectangle (7,7);
\filldraw[fill=yellow] (4,10) rectangle (5,11);
\filldraw[fill=yellow] (6,12) rectangle (7,13);

\coordinate [label=center:{$A$}] (a) at (0.5,0.5);
\coordinate [label=center:{$C$}] (a) at (2.5,2.5);
\coordinate [label=center:{$D$}] (a) at (4.5,4.5);
\coordinate [label=center:{$B$}] (a) at (6.5,6.5);
\coordinate [label=center:{$D'$}] (a) at (4.5,10.5);
\coordinate [label=center:{$B'$}] (a) at (6.5,12.5);

\coordinate [label=center:{$\chi$}] (a) at (-2.5,3.5);
\coordinate [label=center:{$\chi'$}] (a) at (-2.5,9.5);
\end{tikzpicture}}
\end{center}
\caption{Proof of Lemma~\ref{lma:cutBlockOrderPreserve}: the claim that $A \approx^* B'$ in the case where 
$A$ is below $\chi$.}
\label{fig:cutBlockOrderPreserve}
\end{figure}

Recall now that $B'$ is not extremal, and that $a \in A$,
$b \in B_0$ with $a \sim b$. 
Pick any $b' \in B'$. From $a \sim b$ and $B \ll B'$, 
we know that  $b' \not < a$, whence, by Lemma~\ref{lma:cutWeakening},   
$b' \not <^* a$. So suppose, for contradiction, that $a <^* b'$.
By the definition of $<^*$, there exists a sequence
$a = a_0, \dots, a_m = b$ such that, for all $i$ ($0 \leq i < m$),
the pair $\langle a_i, a_{i+1}\rangle$ is in
$(\ll^* \cup <^*_0 \cup <_\times)$.  
But we already know that
$A \approx^* B'$, so there must be some $\ell$ ($0 \leq \ell < m$) such that 
$a_\ell <_\times a_{\ell+1}$. 
Take the largest such value of $\ell$, and let $C$ be the block containing $a_{\ell+1}$. 
Thus,
$C$ is extremal, and indeed,
since $B'$ is non-extremal, we have $\ell < m-1$, and $C \ll^* B'$, whence $C \ll B'$ by Lemma~\ref{lma:cutBlockWeakening}. 
By the fact that $B' = f_{\chi,\chi}(B)$, and $C$ is extremal, we have
$C \ll B$ and hence $C \ll B_0$. On the other hand, $a < a_{\ell+1} \in C$, contradicting the fact that $a \sim b$. Thus, $a \sim^* b'$, as required.
The case where $A$ lies above $\chi'$ proceeds similarly.
\end{proof}

Let us summarize. We started by taking any partial order $\fA = (X, < \tp)$ with factorization $\bbB = (\bB,\ll)$.
We supposed that there existed equivalent cuts $\chi$, $\chi'$ of $\bbB$, with $\chi$ below $\chi'$. We then constructed a new partial order $\fA^* = (X^*, <^* \tp^*)$ with
factorization $\bbB^* = (\bB^*, \ll^*)$,
as established by Lemma~\ref{lma:cutFundamental}.
Let us write $\fA/(\chi, \chi')$ for $\fA^*$ and
$\bbB/(\chi, \chi')$ for $\bbB^*$. Notice that the size of $\bbB/(\chi, \chi')$---i.e. the number of blocks it contains---is strictly smaller than that of $\bbB$. 

\begin{lemma}
Let  $\Psi$ be a conjunction of basic formulas, and suppose
$\fA$ is a typed partial order such that $\fA \models \Psi$.
Let $\bbB$ be a factorization of $\fA$ such that 
$\bbB \models \FC(\Psi)$, and suppose $\chi$, $\chi'$ are equivalent cuts in $\bbB$. Then $\fA/(\chi,\chi') \models \Psi$, 
$\bbB/(\chi,\chi') \models \FC(\Psi)$ and $\fA/(\chi,\chi')$ is thin over $\bbB/(\chi,\chi')$. Moreover, if $\bbB$ is unitary, then so is $\bbB/(\chi,\chi')$.
\label{lma:cutPreserve}
\end{lemma}
\begin{proof}
Write $\fA = (X,<,\tp)$ and $\fA/(\chi,\chi') = (X^*,<^*,\tp^*)$.
We consider the various basic forms in turn.
\begin{description}
	\item \eqref{eq:b1a}, \eqref{eq:b1b}:  $X^* \subseteq X$. 
	\item \eqref{eq:b2a}, \eqref{eq:b2b}, \eqref{eq:b4}: $X^* \subseteq X$ and, by Lemma~\ref{lma:cutWeakening}, $<^*$ is no stronger than $<$. 
	\item \eqref{eq:b3}: By Lemma~\ref{lma:cutPreserveClean}, $\bbB/(\chi,\chi') \models \psi$.
	\item \eqref{eq:b5a}: By Lemma~\ref{lma:cutWeakening}, $<^*$ coincides 
	   with $<$ on $(\tp^*)^{-1}(\alpha)$.
	\item \eqref{eq:b5b}: By Lemma~\ref{lma:cutPreserveAlternating}, $\bbB/(\chi,\chi') \models \psi$.	
	\item \eqref{eq:b6}: Pick any $a \in X^*$ such that $\tp^*(a) = \alpha$.
	Let $a^*$ be a maximal $\alpha$-element of $X$ such that $a < a^*$. Since $\fA \models \psi$, let $b \in X$
	be such that $\fA \models \mu[b]$,
	$\tp[b] \neq \alpha$
	and $a^* < b$. Without loss of generality,
	we may assume that $b$ is a maximal element of its 1-type in $\fA$. Since $a^*$ and $b$ are maximal elements, we have $a^*,b \in X^*$, and indeed $a^* <^* b$. Finally, by the second statement of Lemma~\ref{lma:cutWeakening}, $a <^* a^*$, whence $a <^* b$, whence
	$\fA/(\chi,\chi') \models \psi$.
	\item \eqref{eq:b7}: Proceed symmetrically to the case (B6). 
	\item \eqref{eq:b8}: By Lemma~\ref{lma:cutBlockOrderPreserve}.
	\item \eqref{eq:b9}:  $X^* \subseteq X$. 
	\item \eqref{eq:b10}: All extremal blocks of $\bbB$ are blocks of $\bbB/(\chi,\chi')$. 
\end{description}
Lemmas~\ref{lma:cutPreserveClean} and~\ref{lma:cutPreserveAlternating} ensure that 
$\bbB(\chi,\chi') \models \FC(\Psi)$. The remaining statements of the lemma are obvious.
\end{proof}

\begin{lemma}
Suppose $\Psi$ is a finitely satisfiable conjunction of basic formulas. Then there is a model $\fA \models \Psi$ with unitary factorization $\bbB$ of size bounded by a doubly exponential function of
the size of the signature of $\Psi$, such that $\fA$ is thin over $\bbB$.
\label{lma:fewBlocks}
\end{lemma}
\begin{proof}
Suppose $\fA_0$ is a finite typed partial order such that $\fA_0 \models \Psi$. By Lemma~\ref{lma:factorize}, let $\bbB_0$ be a factorization of $\fA_0$ such that $\bbB_0 \models \FC(\Psi)$.
By Lemma~\ref{lma:loseFiveKG}, we may assume that $\bbB_0$ is unitary and that $\fA_0$ is thin over $\bbB_0$.
Assuming $\fA_i$ and $\bbB_i$ have been defined, if $\bbB_i$ contains a pair of equivalent cuts, $\chi$ and $\chi'$, let $\fA_{i+1} = \fA_i/(\chi,\chi')$ and $\bbB_{i+1} = \bbB_i/(\chi,\chi')$. By 
Lemma~\ref{lma:cutPreserve}, $\fA_{i+1} \models \Psi$ and 
$\bbB_{i+1} \models \FC(\Psi)$; moreover, $\bbB_{i+1}$ is unitary, and $\fA_{i+1}$ is thin over $\bbB_{i+1}$.
Since the number of blocks in $\bbB_i$ is strictly decreasing, we eventually reach a structure $\fA_m$ with factorization $\bbB_m$, in which no two cuts are equivalent. Since the frontier of any cut is at most exponential in size, there can be at most doubly exponentially many cuts in $\bbB_m$, and hence at most 
doubly exponentially many blocks in $\bbB_m$. This proves the lemma.
\end{proof}

\subsection{Reducing the size of blocks}
\label{subsec:blockShrink}
With Lemma~\ref{lma:fewBlocks}, we have established that, if a collection $\Psi$ of basic formulas has a finite model, then it has a finite model $\fA$ with a small factorization $\bbB$, such that 
$\bbB$ guarantees the truth of all factor-controllable members of $\Psi$, and $\fA$ is thin over $\bbB$. However, while the \textit{number} of the blocks in $\bbB$ was bounded by a doubly exponential function of the size of the signature of $\Psi$, nothing at all was said about their \textit{size}.
In this section we show that the blocks themselves can bounded in size. 

Fix some finite typed partial order $\fA = (X,<,\tp)$ with unitary factorization $\bbB = (\bB, \ll)$, such that $\fA$ is thin over $\bbB$. 
Let us suppose that, for some finite set $\Psi$
of basic formulas, $\fA \models \Psi$ and $\bbB \models \FC(\Psi)$.
Our strategy in the sequel will be to divide up the blocks of $\bB$ into
{\em sub-blocks}, and then to replace each sub-block by a set of either one or two elements,
imposing a partial order on these elements which secures satisfaction of $\Psi$.
The difficulty is that,
in reducing the size of each block, we are in danger of creating connections
between elements arising from previously unrelated blocks, and in particular of creating unwanted cycles
in the partial order we are trying to define. A sub-block will be replaced by a singleton if the block that includes it is itself is a unit block; otherwise, it will be replaced by a pair of incomparable elements. The assumption that $\fA$ is thin over $\bbB$ underpins an inductive  argument in Lemma~\ref{lma:subBlockMonotonicity} crucial in showing that 
the order we eventually define contains no cycles. 
The assumption that $\bbB$ is unitary rules out the possibility that some sub-block is made to contain a pair of incomparable elements when the including block is required to be linearly ordered---in particular,
if $\Psi$ contains a basic formula of type~\eqref{eq:b5a}.
 
For $a \in X$, and $B \in \bB$, we say that $B$ is {\em below} $a$ if there exists $b \in B$ such that $b < a$; similarly, we say that $B$ is {\em above} $a$ if there exists $b \in B$ such that $a < b$. Notice that, if
$a \in B$, we may have $B$ above and indeed also below $a$. We define the \textit{sub-type} of $a$ to be the triple
$\langle \bB^-, B, \bB^+ \rangle$, where $B$ is the block containing $a$,
$\bB^-$ is the set of blocks below $a$ and $\bB^+$ the set of blocks above $a$. A {\em sub-block} is a maximal set $s$ of elements all having the same sub-type. We write $\tp(s) = \tp(B)$.
Thus, each block is partitioned into a finite number of sub-blocks, and all elements of a sub-block $s$ have 1-type $\tp(s)$. If $s$ and $t$ are sub-blocks contained in the respective blocks
$A$ and $B$, and $A \ll B$, we write $s \ll t$.

\begin{lemma}
Let $a, b \in X$ with $a < b$. Let the sub-type of $a$ be $s = \langle \bA^-, A, \bA^+ \rangle$ and 
the sub-type of $b$ be $t = \langle \bB^-, B, \bB^+ \rangle$, where $s \neq t$. Then
\textup{(i)} $\bA^- \subseteq \bB^-$;
\textup{(ii)}  $\bA^+ \supseteq \bB^+$; and
\textup{(iii)}  $\set{A} \cup \bA^+ \supseteq \set{B} \cup \bB^+$.
Furthermore, at least one of these inclusions is strict.
\label{lma:subBlockMonotonicity}
\end{lemma}
\begin{proof}
The inclusions themselves are immediate. To show strictness, therefore, since $\fA$ is thin over $\bbB$, it evidently suffices to prove the result in the cases where one of $a<_0 b$, $a \ll b$ or $a <_\times b$ holds. If $a<_0 b$, then $A = B$,
and strictness of either (i) or (ii) follows from the assumption $s \neq t$. If $a \ll b$, then $A \ll B$, and hence $A \not \in \bB^+$, since there certainly cannot exist $a' \in A$ with $b < a'$. Hence, inclusion (iii) is strict. Suppose, then that $a <_\times b$, so that $A \neq B$, and $a$ is either a maximal $\alpha$-element or 
a minimal $\alpha$-element, where $\alpha= \tp(a)$. If the former, then, since $a < b$, we again have $A \not \in \bB^+$, so that inclusion (iii) is strict. If the latter, then $A \not \in \bA^-$, so 
inclusion (i) is strict.
\end{proof}

For every sub-block $s$, let $\hat{s} = \set{\hat{s}(0), \hat{s}(1)}$, where
$\hat{s}(0)$ and $\hat{s}(1)$ are some objects. 
If $s$ is contained in (and hence is equal to) a unit block, we set 
$\hat{s}(0) = \hat{s}(1)$; otherwise, we set $\hat{s}(0) \neq \hat{s}(1)$. Thus,
each $\hat{s}$ has cardinality either 1 or 2. We call objects of
the form $\hat{s}(0)$ \textit{left-objects}, and those
of the form $\hat{s}(1)$ \textit{right-objects}.
For $s \neq t$, we insist that $\hat{s} \cap \hat{t} = \emptyset$. 

For each $B \in \bB$, let 
$\hat{B} = \bigcup \set{\hat{s} \mid \text{$s$ a sub-block of $B$}}$. Now
let $\hat{\bB} = \set{\hat{B} \mid B \in \bB}$, and $\hat{X} = \bigcup \hat{\bB}$. Define an relation $\prec$ on $\hat{X}$ to be the transitive closure of
$r_\exists \cup r_\forall$, where
\begin{align*}
r_\exists & = \{\langle \hat{s}(i),\hat{t}(i) \rangle \mid
\text{$i \in \{0,1\}$, $s \neq t$ and there exist $a \in s$, $b \in t$ such that $a < b$}\}\\
r_\forall & = \set{\langle \hat{s}(i),\hat{t}(j) \rangle \mid s \ll t,\ i,j \in \set{0,1}}
\end{align*}
Note that, when sub-blocks $s$ and $t$ contain elements related by $<$, we relate the corresponding left-objects and the corresponding right-objects by $r_\exists$; however, unless either of $s$ or $t$ has cardinality 1, we do not relate left-objects to right-objects via $r_\exists$ or vice versa. On the other hand, if the block order $\bbB$ enforces an ordering between the elements
of $s$ and $t$, we relate all elements of $\hat{s}$ to all elements of $\hat{t}$ via $r_\forall$. The idea is to keep left-hand elements from being related to right-hand elements by $\prec$ wherever possible. 
  
\begin{lemma}
	If $c \in \hat{s}$, $d \in \hat{t}$ and $c \prec d$, then $\hat{s} \neq \hat{t}$. Hence,
	$\prec$ is a partial order on $\hat{X}$. 
	\label{lma:isPO}
\end{lemma}
\begin{proof}
	Suppose $c = c_0, \dots, c_m = d$ is a sequence of elements of $\hat{X}$, 
	where $\langle c_i, c_{i+1} \rangle \in r_\exists \cup r_\forall$ for all $i$ ($0 \leq i < m$).
	For all $i$ ($0 \leq i \leq m$), let $c_i \subseteq \hat{s}_i$, and let $\langle \bA^-_i, A_i, \bA^+_i 
	\rangle$ be the sub-type defining the sub-block $s_i$. It is immediate from Lemma~\ref{lma:subBlockMonotonicity} that this sequence of sub-types cannot contain repeated elements.
	It immediately follows that $\prec$ is irreflexive. 
\end{proof}
Define the function $\hat{\tp}$ on $\hat{X}$ by setting $\hat{\tp}(\hat{s}(i)) = \tp(s)$ for every sub-block $s$ and every $i \in \set{0,1}$. Now define $\hat{\fA}$ to be the typed partial order $(\hat{X},\prec, \hat{\tp})$.
Note that, if $\hat{B} \in \hat{\bB}$, the 1-type $\hat{\tp}(\hat{s}(i))$ is constant for all $\hat{s}(i) \in \hat{B}$; we
denote this value by $\hat{\tp}(\hat{B})$.
Finally, we define a partial order $\curlyeqprec$ on $\hat{\bB}$ by setting $\hat{A} \curlyeqprec \hat{B}$ just in case $A \ll B$, and define $\hat{\bbB} = (\hat{\bB},\curlyeqprec)$. 

The number of sub-types is bounded by $|\bB|^{2N+1}$, where $N$ is the number of 1-types. To see this, notice that, since blocks of any given type are linearly ordered,
the sets $\bB^-$ and $\bB^+$ are each specified by a sequence of at most $N$ blocks. 
At the same time, $|\hat{X}| \geq 2$. Indeed, if $|\bB| \geq 2$, this is immediate. If, on the
other hand, $\bB = \set{B}$, then $|B| = |X| \geq 2$, whence $|\hat{X}| = |\hat{B}| \geq 2$. 
Thus, $\hat{\fA}$ does not violate our general restriction to structures of cardinality at least 2.
We now prove a sequence of lemmas culminating in Lemma~\ref{lma:crunchTruth}, which states that $\hat{\fA} \models \Psi$. 
\begin{lemma}
	$\hat{\bbB}$ is a factorization of $\hat{\fA}$. Moreover, the mapping $B \mapsto \hat{B}$ is an isomorphism of typed partial-orders $(\bbB, \ll, \tp) \rightarrow (\hat{\bbB}, \curlyeqprec, \hat{\tp})$. 
\label{lma:isomorphicFactorization}
\end{lemma}
\begin{proof}
Immediate from the above construction.
\end{proof}

\begin{lemma}
Let $c \in \hat{s}$ and $d \in \hat{t}$, where $s$, $t$ are sub-blocks of $\bbB$, and let $s \subseteq A$, $t \subseteq B$, where $A$, $B$ are blocks of $\bbB$. If $c \prec d$, then: \textup{(i)} for all $a \in s$, there exists $b \in B$ such that $a < b$; and
\textup{(ii)} for all $b \in t$, there exists $a \in A$ such that $a < b$.  
\label{lma:crunch1}
\end{lemma}
\begin{proof}
We may suppose $c = c_0, \dots, c_m = d$ are elements of $\hat{X}$
such that $(c_i, c_{i+1}) \in r_\exists \cup r_\forall$ for all $i$ ($0 \leq i < m$). We establish (i) by induction on $m$. If $m = 1$, from the definition of $r_\exists$ and $r_\forall$, there exist
$a \in s$ and $b \in t \subseteq B$ 
such that $a < b$. Since $s$ is a sub-block, for all $a \in s$, there exists $b \in B$
such that $a < b$. If $m > 1$, suppose the result holds for $c$, $d$ joined by shorter sequences.
Let $c_1 \in \hat{s}_1$.
From the definition of $r_\exists$ and $r_\forall$, there exist
$a \in s$ and $a' \in s_1$ such that $a < a'$. By inductive hypothesis, there exists $b \in B$ such that
$a' < b$, whence $a < b$. Since $s$ is a sub-block, for all $a \in s$, there exists $b \in B$
such that $a < b$. The proof of (ii) is similar.
\end{proof}
\begin{lemma}
Let $c_1 \in \hat{s}_1$, $c_2 \in \hat{s}_2$ and $c_3 \in \hat{s}_3$, where $s_1$, $s_2$, $s_3$ are sub-blocks of $\bbB$, with $s_2$ included in \textup{(} and hence equal to\textup{)} a unit block $B$. If $s_1 \prec s_2 \prec s_3$, 
then, for all $a \in s_1$ and all $b \in s_3$, $a < b$.  
\label{lma:crunch2}
\end{lemma}
\begin{proof}
By Lemma~\ref{lma:crunch1},
for all $a \in s_1$ there exists $b' \in B$ such that $a < b$, and,
for all $b \in s_3$ there exists $b' \in B$ such that $b' < b$. But $B$ is a singleton, whence $a < b$.
\end{proof}
\begin{lemma}
Let $c_1 \in \hat{s}_1$, $c_2 \in \hat{s}_2$, $c_3 \in \hat{s}_3$ and
$c_4 \in \hat{s}_4$ where $s_1$, $s_2$, $s_3$ $s_4$ are sub-blocks of $\bbB$. If $c_1 \prec c_2$, $\langle c_2, c_3 \rangle \in r_\forall$  and $c_3 \prec c_4$, 
then, for all $a \in s_1$ and all $b \in s_4$, $a < b$.  
\label{lma:crunch3}
\end{lemma}
\begin{proof}
Similar reasoning to Lemma~\ref{lma:crunch2}.
\end{proof}

Now for the crucial lemma guaranteeing the existence of incomparable witnesses in the typed partial order $\hat{\fA}$.
For $c, d \in \hat{X}$, we write $c \asymp d$ to mean that $c \neq d$, $c \not \prec d$ and $d \not \prec c$. That is, $\asymp$ stands in the same relation to $\prec$ as $\sim$ does to $<$.
\begin{lemma}
Suppose $a, b \in X$ with $a \sim b$.
Let $s$ be the sub-block containing $a$ and $t$ the sub-block containing $b$. Then, for every $c \in \hat{s}$, there exists $d \in \hat{t}$ such that 
$c \asymp d$.   
\label{lma:crunchPreserve}
\end{lemma}
\begin{proof}
Assume without loss of generality that $c = \hat{s}(0)$ is a left-element. We claim that $d= \hat{t}(1)$ is incomparable to $c$. For suppose $c \prec d$. Then there is a sequence $c = c_0, \dots, c_m = d$ of elements of $\hat{X}$
such that $(c_i, c_{i+1}) \in r_\exists \cup r_\forall$ for all $i$ ($0 \leq i < m$). Let $c_i \in \hat{s}_i$ and $s_i \subseteq A_i$ for all $i$. Since $c$ is a left-element and $d$ is a right element, either 
$A_i$ is a unit block for some $i$ ($0 \leq i \leq m$) or 
$(c_i, c_{i+1}) \in r_\forall$ for some $i$ ($0 \leq i < m$). It then follows from Lemmas~\ref{lma:crunch2} or~\ref{lma:crunch3} that, for all $a' \in s$,
and all $b' \in t$, $a' < b'$, contradicting the supposition that
$a \sim b$. Hence $c \not \prec d$. By a similar argument,
$d \not \prec c$.
\end{proof}

\begin{lemma}
$\hat{\fA} \models \Psi$.
\label{lma:crunchTruth}
\end{lemma}
\begin{proof}
We consider the possible forms of $\psi \in \Psi$ in turn. 
\begin{description}
\item \eqref{eq:b1a}: $\fA \models \psi$ implies that 
there is just one block $A$ of $\bbB$ having 1-type $\alpha$, and $A = s$ is a unit block. But then there is only one element of $\hat{X}$ having 1-type $\alpha$, namely
$\hat{s}(0) = \hat{s}(1)$. Thus $\hat{\fA} \models \psi$.

\item \eqref{eq:b1b}: This formula is equivalent to $\forall x \neg \alpha \vee  \forall x \neg \beta$; but
the realized 1-types in $\fA$ and $\hat{\fA}$ are the same.
	
\item \eqref{eq:b2a}: Suppose $c \in \hat{s}$, $d \in \hat{t}$ be such that $\hat{\tp}(c) = \alpha$ 
and $\hat{\tp}(d) = \alpha$. Let $s$, $t$ be sub-blocks of the respective blocks $A$ and $B$. If $c \prec d$, then, by Lemma~\ref{lma:crunch1}, there exist $a \in A$ and $b \in B$
such that $a <b$, contradicting $\fA \models \psi$. Similarly if $d \prec c$. Thus, $\hat{\fA} \models \psi$.

\item \eqref{eq:b2b}: Similar to \eqref{eq:b2a}.

\item \eqref{eq:b3}, \eqref{eq:b5b}: By Lemma~\ref{lma:isomorphicFactorization}, $\hat{\bbB}$ is isomorphic to $\bbB$ (as a typed partial order), so that $\hat{\bbB} \models \psi$. But since $\hat{\bbB}$ is a factorization of $\hat{\fA}$, we have,
by the first statements of Lemmas~\ref{lma:lmaBlockAlphaBetaGlobal} and~\ref{lma:lmaBlockAlphaBetaAlternate}, $\hat{\fA} \models \psi$.

\item \eqref{eq:b4}: Let $c= \hat{s}(i)$ be of 1-type $\beta$ and $d= \hat{t}(j)$ be of 1-type $\alpha$. Let $t$ lie in the block $B$ of $\bbB$. Suppose, for contradiction, $c \prec d$. Pick any $a \in s$. By Lemma~\ref{lma:crunch1}, there exists 
$b \in B$ such that $a < b$. But $a$ is of type $\beta$ and $b$ of type $\alpha$, contradicting $\fA \models \psi$. Hence $\hat{\fA} \models \psi$.

\item \eqref{eq:b5a}: $\fA \models \psi$ implies that every block of $\bbB$ having 1-type $\alpha$ is linearly ordered, and
hence, by assumption, is in fact a unit-block. But then 
every block of $\hat{\bbB}$ having 1-type $\alpha$ is a unit-block, whence $\hat{\fA} \models \psi$.

\item \eqref{eq:b6}: Let $c= \hat{s}(i)$ be of 1-type $\alpha$, and pick any $a \in s$. Since $\fA \models \psi$, we have $b > a$ such that $\tp(b) \neq \alpha$ and $\fA \models \mu[b]$. Let $b$ be in the sub-block $t$, and let $d = \hat{t}(i)$. By construction, $c \prec d$
and $\hat{\fA} \models \mu[d]$, whence $\hat{\fA} \models \psi$. 

\item \eqref{eq:b7}: Similar to~\eqref{eq:b6}.

\item \eqref{eq:b8}:
Let $c \in \hat{s}$ be of 1-type $\alpha$, and pick any $a \in s$. Since $\fA \models \psi$, we have $b \sim a$ such that $\fA \models \mu[b]$. Let $b$ be in the sub-block $t$. By Lemma~\ref{lma:crunchPreserve} there exists $d \in \hat{t}$ such that $c \asymp d$. Hence $\hat{\fA} \models \psi$.

\item \eqref{eq:b9}, \eqref{eq:b10}: The realized 1-types in $\fA$ and $\hat{\fA}$ are the same.
\end{description}
\end{proof}

\begin{theorem}
Let $\phi$ be an $\LtoPOu$-formula in weak normal form with multiplicity $m$ over a signature $\sigma$. If $\phi$ has a finite model, then it has a model of size bounded by a doubly exponential function of $|\sigma| +m$.
Hence, any finitely satisfiable $\LtoPOu$-formula $\phi$ has a model 
of size bounded by a doubly exponential function of
$\sizeOf{\phi}$, and so $\FinSat(\LtoPOu)$ is in $\TwoNExpTime$.
\label{theo:mainPOUnary}
\end{theorem}
\begin{proof}
For the first statement, by Lemma~\ref{lma:nfPO}, we may replace $\phi$ by a set $\Psi$ of basic formulas over a signature $\sigma^*$ of size at most $|\sigma| +3m$.
By Lemma~\ref{lma:fewBlocks}, let $\fA$ be a typed partial order with unitary factorization $\bbB = (\bB, \ll)$. such that
$\fA  \models \Psi$, 
$\bbB \models \FC(\Psi)$, $\bbB$ is of size doubly exponential in $|\sigma^*|$, and $\fA$ is thin over $\bbB$. 
Now let $\hat{\fA}$ be as defined before Lemma~\ref{lma:isPO}. By Lemma~\ref{lma:crunchTruth}, $\hat{\fA} \models \Psi$. But $\hat{\fA}$ is of size at most $2(|\bB|^{2N+1})$, where $N$ is the number of 1-types over $\sigma^*$. Thus, $\Psi$ is satisfiable over a domain doubly exponential in $|\sigma|+m$.
The remainder of the theorem follows by Lemma~\ref{lma:nfL2}.
\end{proof}

\section{Two-variable logic with one partial order}
\label{sec:L21po}
The purpose of this section is to show that the logic $\LtoPO$ has the doubly 
exponential-sized finite model property (Theorem~\ref{theo:mainPO}): if $\phi$ is a finitely satisfiable $\LtoPO$-formula, then $\phi$ has a model
of size bounded by some fixed doubly exponential function of $\sizeOf{\phi}$. It follows that the finite satisfiability problem for $\LtoPO$ is in \TwoNExpTime. We proceed by reduction to the
corresponding problem for weak normal-form $\LtoPOu$-formulas, paying particular attention to the size of the relevant signature, and the multiplicities of the formulas in question.
In this section, we continue to assume that all signatures contain the navigational  predicates $<$, $>$ and $\sim$, subject to the usual semantic constraints. We use the (possibly decorated) variable $\tau$ to range over 2-types, $\lambda$, $\mu$, $\nu$ over unary pure Boolean formulas
and $\zeta$, $\eta$, $\theta$, $\phi$, $\chi$, $\psi$, $\omega$ over arbitrary formulas.
Henceforth, for any integer $n$, we denote by $\lfloor n \rfloor$ the value $n$ modulo 3. 

A crucial step in our reduction is the definition of a specialized normal form for $\LtoPO$-formulas, from which it is easy to eliminate ordinary binary predicates. 
Say that an $\LtoPO$-formula is in {\em spread normal form} if it conforms to the pattern
\begin{equation}
\begin{split}
\bigwedge_{\zeta \in Z} \exists x . \zeta \wedge &
 \forall x \forall y (x = y \vee \eta) \wedge\\ 
&       \bigwedge_{k=0}^2 \bigwedge_{h=0}^{m-1} \forall x \exists y (\lambda_k \rightarrow  (\lambda_{\lfloor k+1 \rfloor}(y) \wedge \mu_h(y) \wedge \theta_h)), 
\end{split}
\label{eq:spreadNf}
\end{equation}
where: (i) $Z$ is a set of unary pure Boolean formulas; 
(ii) $\eta,  \theta_0, \dots,\theta_{m-1}$ are quantifier- and equality-free formulas, with $m \geq 1$; and (iii) 
$\lambda_0, \lambda_1, 
\lambda_2$ are mutually exclusive unary pure Boolean formulas;
and (iv) 
$\mu_0, \dots, \mu_{m-1}$ are mutually exclusive unary pure Boolean formulas. 
Spread normal form is---modulo insertion of harmless conjuncts $x \neq y$---a special case
of weak normal form~\eqref{eq:nf}. We take the {\em multiplicity} of the spread normal form formula~\eqref{eq:spreadNf}
to be the quantity $3m$. (Thus, the definitions of multiplicity for spread normal form 
and weak normal form agree.) Its distinguishing feature is that 
witnesses are required to be `spread' over disjoint sets of elements. Thus, suppose $\fA$ is a model
of the formula~\eqref{eq:spreadNf}, and 
$\fA \models \lambda_k[a]$ for some $a \in A$ and some $k$ ($0 \leq k < 3$). Then
there exist $b_0, \dots, b_{m-1} \in A$ such that, for
each $h$ ($0 \leq h < m$), $\fA \models \theta_h[a,b_h]$ and $\models \mu_h[b_h]$.
It follows that the $b_0, \dots, b_{m-1}$ are distinct; moreover, 
all of these elements satisfy $\lambda_{\lfloor k+1 \rfloor}(y)$, so that
{\em their} witnesses, which satisfy $\lambda_{\lfloor k+2 \rfloor}(y)$, cannot include $a$. Thus, the witnesses for an element of $\fA$ are never duplicated, and nothing is a witness of a witness of itself. 

In order to transform $\LtoPO$-formulas into spread normal form, we must first
establish a lemma allowing us to create copies of certain parts of structures without compromising
the truth of those formulas.
If $\fA$ is any structure interpreting a signature $\sigma$, we call any element of $a$ a {\em king} if it is the unique element of $A$ realizing its 1-type (over $\sigma$): $\tp^\fA[b] = \tp^\fA[a]$ implies $b = a$ for all $b \in A$. Elements which are kings are said to be {\em royal}.
The following lemma says that we may duplicate the non-royal elements of any structure 
any (finite) number of times.
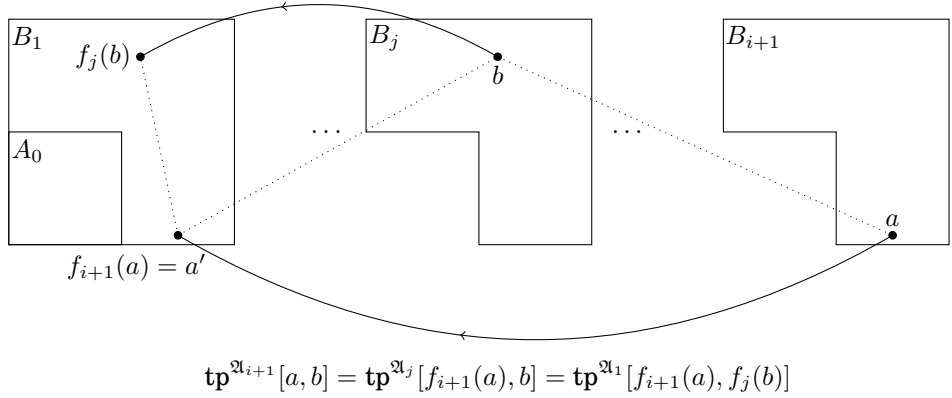
\begin{figure}
\begin{center}
\begin{tikzpicture}[scale= 0.5]

\draw[dotted] (4.5,3.75) -- (3.5,8.5);
\draw[dotted] (4.5,3.75) -- (13,8.5);
\draw[dotted] (13,8.5) -- (23.5,3.75);

\filldraw (0.5,9) node (B1) {$B_1$};
\filldraw (0.5,6) node (A0) {$A_0$};
\draw (0,3.5) rectangle (6,9.5);
\draw (0,3.5) rectangle (3,6.5);
\filldraw (3.5,8.5) node[left] {$f_j(b)$} circle (0.1);
\filldraw (4.5,3.75) circle (0.1);
\filldraw (3.4,3.5) node[below] {$f_{i+1}(a) = a'$};

\draw (8.5,6.5) node {$\dots$};

\filldraw (10,9) node (Bi) {$B_j$};
\draw (9.5,9.5) -- (15.5,9.5) -- (15.5,3.5)  -- (12.5,3.5) -- (12.5, 6.5) -- (9.5,6.5) -- cycle;
\filldraw (13,8.5) node[below] {$b$} circle (0.1);

\draw (16.5,6.5) node {$\dots$};

\filldraw (19.8,9) node (Bj) {$B_{i+1}$};
\draw (19,9.5) -- (25,9.5) -- (25,3.5)  -- (22,3.5) -- (22, 6.5) -- (19,6.5) -- cycle;
\filldraw (23.5,3.75) node[above]{$a$} circle (0.1);

\draw [->-=0.6, bend right] (13,8.5) to (3.5,8.5);
\draw [->-=0.6, bend left] (23.5,3.75) to (4.5,3.75);

\draw (13,0) node {$\tp^{\fA_{i+1}}[a,b] = \tp^{\fA_j}[f_{i+1}(a),b] = \tp^{\fA_1}[f_{i+1}(a),f_j(b)]$};
\end{tikzpicture}
\end{center}
\caption{Construction of the structure $\fA_{i+1}$ (Lemma~\ref{lma:betterCopyLemma}), where $a \in B_{i+1}$, $b \in B_j$, and $j \leq i$.}
\label{fig:betterCopyLemma}
\end{figure}
\begin{lemma}
Let $\fA_1$ be a structure over domain $A_1$,
$A_0$ the set of kings of $\fA_1$, and $B_1 = A_1 \setminus A_0$. 
There exists a family of sets $\set{B_i}_{i \geq 2}$, pairwise disjoint and disjoint from $A_1$, 
a family of bijections
$\set{f_i}_{i \geq 2}$, where $f_i: B_i \rightarrow B_1$,
and a sequence of structures $\set{\fA_i}_{i \geq 2}$, where
$\fA_i$ has domain $A_i = A_0 \cup B_1 \cup B_2 \cup \cdots \cup B_i$, such that, for all $i \geq 2$:
\begin{enumerate}[\textup{(}i\textup{)}]
\item $\fA_{i-1} \subseteq \fA_i$, and all 2-types realized in $\fA_i$ are realized in $\fA_{1}$;
\item for all $a \in B_i$ and all $b \in A_1$, if $f_i(a) \neq b$, then
$\tp^{\fA_{i}}[a,b] = \tp^{\fA_{1}}[f_i(a),b]$;
\item for all $a \in B_i$, all $j$ \textup{(}$2 \leq j \leq i$\textup{)} and all $b \in B_j$, if $f_i(a) \neq f_j(b)$, then %then \newline
$\tp^{\fA_{i}}[a,b] = \tp^{\fA_1}[f_i(a),f_j(b)]$;
\item $<^{\fA_i}$ is a partial order.
\end{enumerate}
\label{lma:betterCopyLemma}
\end{lemma}
\begin{proof}
Enumerate  $B_1$ as $\set{a^1, a^2, a^3, \dots}$. Let the set of indices of this enumeration (which may be finite or infinite) be $K$. For each $k \in K$, let $b^k \in B_1$ be such that $\tp^{\fA}[a^k]= \tp^{\fA}[b^k]$ but $a^k \neq b^k$. This is possible because $A_0$ is the set of kings of $\fA_1$. 

We prove the lemma by induction on $i$. The base case, $i =1$, is vacuous. Observe that the domain of $\fA_1$ is
the disjoint union of $A_0$ and $B_1$, and define the bijection $f_1: B_1 \rightarrow B_1$ to be the
identity map.
For the inductive case, suppose the sets $A_{i-1}$, $B_i$, the structure $\fA_i$, and the bijection $f_i: B_i \rightarrow B_1$ have been defined, such that the domain $A_i$ of $\fA_i$ is
the disjoint union of $A_{i}$ and $B_i$, and 
Statements (i)--(iv) hold whenever $i \geq 2$. 
We proceed to define $A_{i+1}$, $B_{i+1}$, $\fA_{i+1}$ and $f_{i+1}$, and establish the corresponding properties for these objects.

The definition employs a subsidiary induction.
Let $A^1_i = A_i$ and $\fA^1_i = \fA_i$. We shall construct a sequence of
structures $\set{\fA^k_i}_{k \in K}$ over the corresponding sequence of domains $\set{A^k_i}_{k \in K}$.
Assume that $\fA^k_i$ has been defined over domain $A^k_i$ and $k+1 \in K$. Let $a^k_{i+1}$ be a new element (not in $A^k_i$), 
and let $A^{k+1}_i = A^k_i \cup \set{a^k_{i+1}}$. (The indexing reflects the intuition that $a^{k}_{i+1}$ will form the $k$th new element in the structure $\fA_{i+1}$ when this is completed.) We extend $\fA^k_i$ to a structure $\fA^{k+1}_i$ over $A^{k+1}_i$ 
by setting:
\begin{align}
\tp^{\fA^{k+1}_i}[a^k_{i+1}] & = \tp^{\fA^{k}_i}[a^k] \label{eq:innerBuild1}\\
\tp^{\fA^{k+1}_i}[a^k_{i+1},a^k] & = \tp^{\fA^{k}_i}[b^k,a^k]\label{eq:innerBuild2}\\
\tp^{\fA^{k+1}_i}[a^k_{i+1},b] & = \tp^{\fA^{k}_i}[a^k,b] \quad \text{for all $b \in A^k_i \setminus \set{a^k}$}. \label{eq:innerBuild3}
\end{align}
From the fact that $\tp^{\fA^k_i}[a^k] = \tp^{\fA^{k}_i}[b^k]$, these type-assignments involve no clashes. Moreover, since $\fA^1_i \subsetneq \fA^2_i \subsetneq \cdots$, we may define $\fA_{i+1} = \bigcup_{k \in K} \fA^k_i$, taking $\fA_{i+1}$ to have domain $A_{i+1}$. 
Letting $B_{i+1} = \set{a^1_{i+1}, a^2_{i+1}, \dots}$, we see that $A_{i+1} = A_{i} \cup B_{i+1} = A_0 \cup B_0 \cup \cdots \cup B_{i+1}$. Intuitively,
$\fA_{i+1}$ is just like $\fA_{i}$ except that 
we have added an extra copy of the set $B_1$, relating the new elements to each other and to $\fA_i$ as specified by $\fA_1$. Define the bijection 
$f_{i+1}: B_{i+1} \rightarrow B_1$ by setting $f_{i+1}(a^k_{i+1}) = a^k$ for all $k \in K$. 

We need to secure Statements (i)--(iv) of the lemma, but with $i$ replaced by $i+1$.
For Statement~(i), it is immediate by construction that $\fA_i \subseteq \fA_{i+1}$ and from~\eqref{eq:innerBuild2}
and~\eqref{eq:innerBuild3}, via a subsidiary induction on $k$, we see that $\fA_{i+1}$ realizes only those 2-types realized in $\fA_i$, and hence, by inductive hypothesis, in $\fA_1$.
For Statement~(ii), it 
follows from~\eqref{eq:innerBuild3}, again via a subsidiary induction on $k$, that, 
for all $a \in B_{i+1}$ and all $b \in A_1$, if $f_{i+1}(a) \neq b$, then
$\tp^{\fA_{i+1}}[a,b] = \tp^{\fA_1}[f_{i+1}(a),b]$. 
For Statement (iii), we consider separately the cases $j = i+1$ and $j \leq i$.
The former is the simpler: observe that, for $k, \ell \in K$ with $k < \ell$, 
$\tp^{\fA^{\ell}_i}[a^k_{i+1},a^{\ell}_{i+1}] = \tp^{\fA_1}[a^k,a^{\ell}]$. Indeed,
$\tp^{\fA^{\ell}_i}[a^k_{i+1},a^{\ell}_{i+1}] = \tp^{\fA^{\ell-1}_i}[a^k_{i+1},a^\ell] = \tp^{\fA^{k}_i}[a^k_{i+1},a^\ell]
= \tp^{\fA^{k-1}_i}[a^k,a^\ell]= \tp^{\fA_1}[a^k,a^\ell]$. Thus, for distinct $a,b  \in B_{i+1}$, we have
\begin{equation*}
\tp^{\fA_{i+1}}[a,b] = \tp^{\fA_1}[f_{i+1}(a), f_{i+1}(b)].
\end{equation*}
This secures Statement (iii) for the case $j = i+1$.
The case $2 \leq j \leq i$ is illustrated in Fig.~\ref{fig:betterCopyLemma}. 
Writing $a' = f_{i+1}(a)$, by statement (ii) of the inductive hypothesis, if $f_j(b) \neq a'$, then
$\tp^{\fA_i}[a',b] = \tp^{\fA_1}[a',f_j(b)]$, and by construction of $\fA_{i+1}$, 
$\tp^{\fA_{i+1}}[a,b] = \tp^{\fA_{i}}[a',b]$. Thus, if $f_{i+1}(a) \neq f_j(b)$, then 
$\tp^{\fA_{i+1}}[a,b] = \tp^{\fA_{1}}[f_{i+1}(a),f_j(b)]$.

Turning to Statement (iv), it follows from~\eqref{eq:innerBuild1} that $<^{\fA^k_{i+1}}$ is not reflexive. We claim that,
in addition, this relation is transitive.
It evidently suffices to show that if $<^{\fA^k_{i}}$ is a transitive relation, then so is $<^{\fA^{k+1}_{i}}$.
Suppose, therefore that $<^{\fA^k_{i}}$ is transitive, and let $a$, $b$, $c$ be elements of $A^{k+1}_i$ such that
$\fA^{k+1}_i \models a < b$ and $\fA^{k+1}_i \models b < c$. We must show $\fA^{k+1}_i \models a < c$.
If $a, b, c \in A^{k}_i$, this is immediate. Moreover, if $a = b$ or $b = c$ there is nothing to show. 
On the other hand, if $a = c = a^k_{i+1}$ and $b \in A^{k}_i$, then either
$\tp^{\fA^{k+1}_i}[a,b] = \tp^{\fA^{k+1}_i}[c,b] = \tp^{\fA^{k}_i}[a^k,b]$ or 
$\tp^{\fA^{k+1}_i}[a,b] = \tp^{\fA^{k+1}_i}[c,b] = \tp^{\fA^{k}_i}[b^k,b]$, in either case contradicting the supposition
that $\fA^{k+1}_i \models a < b$ and $\fA^{k+1}_i \models b < c$. Moreover, an exactly similar
argument applies if $a = c \in A^k_i$ and $b= a^{k}_{i+1}$. Hence, we may assume that
the elements $a$, $b$ and $c$ are distinct, and that exactly one of them is equal to $a^k_{i+1}$. We therefore have three cases to consider. 

\bigskip

\noindent
Case 1: $a^k_{i+1} = a$. We claim first of all that $c \neq a^k$. For suppose $c = a^k \neq b$. Then 
$\fA^{k+1}_i \models b < c$ implies $\fA^{k}_i \models b < a^k$, whence $\fA^{k}_i \not \models a^k < b$, 
and therefore $\fA^{k+1}_i \not \models a^k_{i+1} < b$,
contradicting the supposition that $\fA^{k+1}_i \models a < b$. If, on the other hand, $b = a^k \neq c$, then $\fA^{k}_i \models b < c$ is the statement $\fA^{k}_i \models a^k < c$, which implies $\fA^{k+1}_i \models a < c$. Thus we may suppose that neither $b$ nor $c$ is equal to $a^k$. But then 
$\fA^{k}_i \models a^k < b$ and $\fA^{k}_i \models b < c$, whence $\fA^{k}_i \models a^k < c$, whence
$\fA^{k+1}_i \models a < c$. 

\bigskip

\noindent
Case 2: $a^k_{i+1} = b$. Suppose first that $a = a^k$. Then $\fA^{k+1}_i \models b < c$ implies 
$\fA^{k}_i \models a^k < c$, and hence $\fA^{k+1}_i \models a^k < c$, which is the required statement
$\fA^{k+1}_i \models a < c$. A similar argument applies if $c = a^k$. 
Thus we may suppose that neither $a$ nor $c$ is equal to $a^k$.
But then 
$\fA^{k}_i \models a < a^k$ and $\fA^{k}_i \models a^k < c$, whence $\fA^{k}_i \models a < c$, whence
$\fA^{k+1}_i \models a < c$.

\bigskip

\noindent
Case 3: $a^k_{i+1} = c$. The same as Case~1, but with the order reversed.

\bigskip
 
\noindent
This completes the induction.
\end{proof}

We now come to the lemma allowing us to transform any $\LtoPO$-formula in standard normal form into one in spread normal form.
We require some additional notation.  Let $ \bar{p} = p_1, \dots, p_n$ be
a sequence of unary predicates. For all $i$ ($0 \leq i < 2^n$), we abbreviate by $\bar{p}\langle i \rangle$ 
the unary, pure Boolean formula $\rho_1 \wedge \cdots \wedge \rho_n$, where, for all $j$
($1 \leq j \leq n$), $\rho_j$ is $p_j(x)$ if the $j$th bit in the $n$-digit binary representation of $i$ is 1, and $\neg p_j(x)$ otherwise. We call
$\bar{p}\langle i \rangle (x)$ the $i$th \textit{labelling formula} (\textit{over} $p_1, \dots, p_n$).
Evidently, if $A = \set{a_0, \dots, a_{M-1}}$ is a set of cardinality $M \leq 2^n$, then we can interpret
the predicates in $p_j$ ($1 \leq j \leq n$) over $A$ so as to ensure that, for all $i$ ($0 \leq i < M$),
$a_i$ satisfies $\bar{p}\langle i \rangle$. 

\begin{lemma}
Let $\phi$ be an $\LtoPO$-formula in standard normal form over a signature $\sigma$, having  multiplicity $m$.
There exists a formula $\phi^*$ in spread normal form over a signature $\sigma^*$ 
with the following properties:
\textup{(i)} $\models \phi^* \rightarrow \phi$; \textup{(ii)} if $\phi$ has a \textup{(}finite\textup{)} model then so has $\phi^*$; and \textup{(iii)} $|\sigma^*|$ is polynomially bounded as a function of  $|\sigma| + m$, and $\phi^*$ has multiplicity $3m$.
\label{lma:asymmetric}
\end{lemma}
\begin{proof}
Write $\phi$ as
\begin{equation*}
\forall x \forall y (x = y \vee \eta) \wedge \bigwedge_{h=0}^{m-1} \forall x \exists y (x \neq y \wedge \theta_h).
\end{equation*}
Suppose $\fA_1 \models \phi$, and let $A_0$ be the set of kings of $\fA_1$. By taking $\fA_1$ to interpret two fresh unary predicates if necessary, we may assume $|A_0| \geq 2$. 
Let $B_1 = A_1 \setminus A_0$ and
let $f_1: B_1 \rightarrow B_1$ be the identity map. 
Now take $\set{B_i}_{i \geq 2}$, $\set{f_i}_{i \geq 2}$ and $\set{\fA_i}_{i \geq 2}$
to be the series of sets, bijections and structures guaranteed by Lemma~\ref{lma:betterCopyLemma}. 
Let  
$\fA = \fA_{3m}$; thus, $\fA$ is finite if $\fA_1$ is. Finally,
re-index the sets $B_1, \dots, B_{3m}$ (in any order whatever) as
$B_{h,k}$, where $0 \leq h < m$ and $0 \leq k < 3$; and re-index the $f_1, \dots, f_{3m}$ correspondingly as$f_{h,k}$.

For each
$a \in A_0$ and each $h$ ($0 \leq h < m$) choose some $b \in A_1 \setminus \set{a}$ such that $\fA_1 \models \theta_h[a,b]$, and let $C_0$ consist of the elements of $A_0$ together with all of the (at most $m \cdot |A_0|$) elements thus selected. We refer to $C_0$ as the \textit{court} of $\fA$. Let us enumerate
$A_0$ as $c_0, \dots, c_{S-1}$ and the rest of $C_0$ as $c_{S}, \dots, c_{T-1}$. Thus, $0 \leq S \leq T 
\leq (m+1)2^{|\sigma|}$. Let $t = \lceil \log (T +1)\rceil $, and let $q_1, \dots, q_t$ be new unary predicates. Writing $\bar{q}\langle i \rangle$ for the $i$th labelling formula over $q_1, \dots, q_t$,
let $\fA$ be expanded to a structure $\fA'$ such that, for all $i$ ($0 \leq i < T$), $\fA' \models  \bar{q}\langle i \rangle [c_i]$, and $\fA' \models  \bar{q}\langle T \rangle [a]$ for all $a \in A \setminus C_0$.

Thus, under the interpretation $\fA'$, for $0 \leq i < S$, we may read $\bar{q}\langle i \rangle(x)$ as ``$x$ is the
$i$th king;'' and
for $0 \leq i < T$, we may read $\bar{q}\langle i \rangle$ as ``$x$ is the
$i$th member of the court.'' (Hence, the kings come before the non-royal courtiers in the numbering.) Now let $\chi$ be the formula 
\begin{equation*}
\bigwedge_{i= 0}^{T-1}\exists x . \bar{q}\langle i \rangle (x) 
\end{equation*}
and $\psi_1$ the formula
\begin{equation*}
\bigwedge_{i=0}^{T-2} \bigwedge_{j= i+1}^{T-1} \forall x \forall y (x = y \vee (\bar{q}\langle i(x) \rangle \wedge \bar{q}\langle j \rangle (y)) \rightarrow \tp^\fA[c_i,c_j]),
\end{equation*}
recording the diagram of $\fA$ over $C_0$. Obviously, $\fA' \models \chi \wedge \psi_1$. 
Conversely, in any model of $\chi \wedge \psi_1$, we see that for all $h$ ($0 \leq h < m$) and
for any element $a$ satisfying
$\bar{p}\langle i \rangle$ for some $i$ ($0 \leq i <S$), there exists 
$b \neq a$ such that the pair $\langle a, b \rangle$ satisfies $\theta_h$. 

Let $s = \lceil \log (S+1) \rceil$. For each $h$ ($0 \leq h < m$), let 
$q^h_1, \dots, q^h_s$ be new unary predicates, and write $\bar{q}^h\langle i \rangle (x)$ for the $i$th labelling formula over these predicates. Expand $\fA'$ to a model $\fA''$ as follows. For each
$a \in A \setminus A_0$, and each $h$ ($0 \leq h < m$), if there exists any $b \in A_0$ such that $\fA \models \theta_h[a,b]$, choose some such element, say, $c_i$ (with $i$ depending on $a$ and $h$), and interpret the
predicates $q^h_1, \dots, q^h_s$ so that
$\fA'' \models  \bar{q}^h\langle i \rangle [a]$; otherwise, interpret the
predicates $q^h_1, \dots, q^h_s$ so that
$\fA'' \models  \bar{q}^h\langle S \rangle [a]$.
Thus, under the interpretation $\fA''$, for $0 \leq i < S$, we may read $\bar{q}^h\langle i \rangle(x)$ as ``$x$ is an element such that the
$i$th king provides a $\theta_h$-witness for $x$.'' Now let $\psi_2$ be the formula 
\begin{equation*}
\bigwedge_{i=0}^{S-1} \bigwedge_{h=0}^{m-1} \forall x \forall y (x = y \vee 
(\bar{q}^h\langle i \rangle(x) \wedge \bar{q}\langle i \rangle (y) \rightarrow \theta_h)),
\end{equation*}
recording this fact. Obviously, $\fA'' \models \psi_2$.
Conversely, in any model of $\chi \wedge \psi_2$, we see that
for all $h$ ($0 \leq h < m$), and all elements
$a$ satisfying $\bar{q}^h\langle i \rangle (x)$ for some $i$ ($0 \leq i < S$),  
there exists
$b \neq a$ such that the pair $\langle a, b \rangle$ satisfies $\theta_h$. 

Finally, let $o_0, o_1, o_2$ and $p_0, \dots, p_{m-1}$ be new unary predicates, and
expand $\fA''$ to a structure $\fA'''$ by setting
\begin{align*}
(o_k)^\fA = & \bigcup_{h=1}^m B_{h,k}  \qquad \qquad  \text{for all $k$ ($0 \leq k <3$)}\\
(p_h)^\fA = & \bigcup_{k=0}^3 B_{h,k} \qquad \qquad  \text{for all $h$ ($1 \leq h \leq m$)}.
\end{align*}
Thus, we may read $o_k(x)$ as ``$x$ is in $B_{h,k}$ for some $h$'', and 
$p_h(x)$ as ``$x$ is in $B_{h,k}$ for some $k$''.
Let $\lambda_0 = o_0(x)$, $\lambda_1 = o_1(x) \wedge \neg o_0(x)$, $\lambda_2 = o_2(x) \wedge
\neg o_0(x) \wedge \neg o_1(x)$. Thus, $\lambda_0, \lambda_1, \lambda_2$ are mutually exclusive pure unary formulas. Similarly, let
$\mu_h(x) = q_h(x) \wedge \bigwedge_{h' = 0}^{h-1} \neg q_{h'}(x)$ for all $h$ ($0 \leq h < m$). Thus,
$\mu_0, \dots, \mu_{m-1}$ are also mutually exclusive unary pure Boolean formulas.

Now let $\psi_3$ be the formula 
\begin{equation*}
\forall x \forall y \left(x = y \vee \bigvee_{i=0}^{S-1} \bar{q}\langle i \rangle (x) \vee \bigvee_{k=0}^2 \lambda_k \right),
\end{equation*}
which, we note, is equivalent (over structures with cardinality at least 2) to
\begin{equation*}
\forall x \left(\bigvee_{i=0}^{S-1} \bar{q}\langle i \rangle (x) \vee \bigvee_{k=0}^2 \lambda_k \right).
\end{equation*}
It is immediate by construction that $\fA''' \models \psi_3$, since every $a \in A_0$ satisfies
$\bar{q}\langle i \rangle(x)$ for some $i$ ($0 \leq i < S$), and  every $a \in A \setminus A_0$ lies
in one of the sets $B_{h,k}$. In addition, let $\theta^*_h(x,y)$ be the formula
\begin{equation*}
\left(\bigwedge_{i=0}^{S-1} \neg \bar{q}^h\langle i \rangle (x)\right) \rightarrow \theta_h,
\end{equation*}
for all $h$ ($0 \leq h < m$), and let $\omega$ be the formula
\begin{equation*}
\bigwedge_{h=0}^{m-1} \bigwedge_{k=0}^2 
\forall x (\lambda_k \rightarrow
\exists y(\lambda_{\modt{k+1}}(y) \wedge  \mu_h(y) \wedge \theta^*_h)).
\end{equation*}
We claim that $\fA''' \models \omega$. To see this, fix $0 \leq h < m$ and $0 \leq k <3$, and
suppose $a \in A$ is such that $\fA''' \models \lambda_k[a]$. If $\fA''' \models \bar{q}^h\langle i \rangle[a]$ for some $i$ ($0 \leq i < S$), then we may pick any element $ b \in B_{h,\lfloor k+1 \rfloor}$ as a witness,
since $\fA''' \models \theta^*_h[a,b]$ holds by failure of the antecedent.
Otherwise, by the construction of $\fA'''$,
$a \in B_{h',k}$ for some $h'$ ($1 \leq h' \leq m$) and, moreover, there is no $b \in A_0$ for which $\fA \models \theta_h[a,b]$.
Now let $a' = f_{h' ,k}(a)$. Since $\tp^\fA[a,b] = \tp^\fA[a',b]$ for all $b \in A_0$, it follows that 
there is no $b \in A_0$ for which $\fA \models \theta_h[a',b]$. Since $\fA_1 \models \phi$, therefore,
let $b' \in B_1$ be such that $\fA \models \theta_h[a',b']$ and let $b \in B_{h,\modt{k+1}}$ be such that 
$f_{h,\modt{k+1}}(b) = b'$. Since $\tp^\fA[a,b] = \tp^\fA[a',b']$, we have $\fA \models \theta_h[a,b]$.
Moreover, by the construction of $\fA'''$, $\fA''' \models \lambda_{\modt{k+1}}[b]$ and $\fA''' \models \mu_h[b]$. Therefore,
$\fA''' \models \omega$ as claimed.
Conversely, in any model of $\omega$, we see that
for all $h$ ($0 \leq h < m$) and all elements $a$
satisfying $\lambda_k(x)$ but 
not satisfying $\bar{q}^h\langle i \rangle (x)$ for any $i$ ($0 \leq i < S$), there exists 
some $b \neq a$ such that the pair $\langle a, b \rangle$ satisfies $\theta_h$. 

Finally, let $\phi^*$ be the formula 
\begin{equation*}
\chi \wedge 
\forall x (x=y \vee \eta) \wedge \psi_1 \wedge \psi_2 \wedge \psi_3 \wedge \omega,
\end{equation*}
and let $\sigma^*$ be the signature of $\psi^*$. 
Thus, $\phi^*$ is in spread form, with multiplicity $3m$.
Moreover, the only new predicates in $\sigma^*$ are
$o_0, o_1, o_2$, $p_0, \dots, p_{m-1}$, $q_1,\dots q_t$, and the $q^h_1, \dots q^h_s$ ($0 \leq h < m$), 
so that
$|\sigma^*|$ is bounded by a polynomial
function of $|\sigma| +m$. Moreover, we have shown that, if $\fA \models \phi$, then 
$\fA''' \models \phi^*$, 
and, moreover, $\fA'''$ is finite if $\fA$ is. It remains to show that $\models 
\phi^* \rightarrow \phi$. So suppose $\fB \models \phi^*$, $a \in B$ and $0 \leq h < m$. 
As we have observed, if $\fB \models 
\bar{q}\langle i \rangle[a]$ for some $i$ ($0 \leq i < S$), then 
$\chi \wedge \psi_1$ guarantees the existence of some 
$b \in B \setminus \set{a}$ such that $\fB \models \theta_h[a,b]$. Otherwise, by $\psi_3$,
$\fB \models \lambda_k[a]$ for some $k$ ($0 \leq k < 3$). If, now $\fB \models 
\bar{q}^h\langle i \rangle[a]$ for some $i$ ($0 \leq i < S$),
$\chi \wedge \psi_2$ guarantees the existence of some $b \in B \setminus \set{a}$ such that $\fB \models \theta_h[a,b]$.
If, on the other hand,
$\fB \not \models \bar{q}^h\langle i \rangle[a]$ for any $i$ ($0 \leq i < S$),
$\omega$ guarantees the existence of some $b \in B \setminus \set{a}$ such that $\fA_1 \models \theta_h[a,b]$. 
Thus, $\fB \models \phi$.
\end{proof}

In the sequel, we employ terminology and techniques familiar from the area of automated theorem proving. In
particular, a {\em a clause} is a disjunction (possibly empty) of literals. The empty disjunction is written as $\bot$, and is taken to denote the falsum.
We use (possibly decorated) lower-case Greek letters $\gamma$, $\delta$, $\epsilon$ to range over clauses, and 
upper-case Greek letters $\Gamma$, $\Delta$ to range over finite sets of clauses. 
If $\Gamma$ is
a finite set of clauses, then we denote by $\Gamma^{-1}$ the result of transposing the variables $x$ and $y$ in $\Gamma$.
To avoid notational clutter, we frequently identify a finite set of clauses with its conjunction, writing, for example, $\Gamma$
when we actually mean $\bigwedge \Gamma$. It is a familiar fact that,
for any quantifier-free formula $\phi$ over relational signature, there exists a collection of clauses $\Gamma$ such that $\models \phi \leftrightarrow \Gamma$ (so-called conjunctive normal form). In general $|\Gamma|$ will be exponential in $\sizeOf{\phi}$; however, $\Gamma$ and $\phi$ employ the same signature. 

Let $\rho$ be an ordinary atomic formula featuring two distinct variables--i.e.~a formula of either of the forms $r(x,y)$ or $r(y,x)$, where $r$ is an ordinary binary predicate, and let $\gamma'$, $\delta'$ be clauses. 
Then, $\gamma = \rho \vee \gamma'$ and $\delta = \neg \rho \vee \delta'$ are also clauses,
as indeed is $\gamma' \vee \delta'$. 
In that case, we call $\gamma' \vee \delta'$ an {\em ordinary binary resolvent} of $\gamma$ and $\delta$, and we say that $\gamma' \vee \delta'$ is {\em obtained by ordinary binary resolution} from $\gamma$ and $\delta$ \textit{on} $\rho$, or simply: $\gamma$ and $\delta$ {\em resolve to form} $\gamma' \vee \delta'$.
Note that no unification of variables occurs in ordinary binary resolution: in fact, ordinary binary resolution
is just the familiar rule of propositional resolution restricted to the case where the resolved-on atom is of the form
$r(x,y)$  or $r(y, x)$, with $r$ an ordinary binary predicate. Observe that: 
(i) if $\gamma$ and $\delta$ resolve to form $\epsilon$, then $\models \gamma \wedge \delta \rightarrow \epsilon$;
(ii) the ordinary binary resolvent of two clauses may or may not
involve ordinary binary predicates; 
(iii) if the clause $\gamma$ involves no ordinary binary predicates, then it cannot undergo ordinary binary resolution at all.

If $\Gamma$ is a set of clauses, denote by $[\Gamma]^*$ the smallest set of clauses including $\Gamma$ and closed under ordinary binary resolution, in the sense that, if $\gamma, \delta \in [\Gamma]^*$, and $\epsilon$ is an ordinary binary resolvent of $\gamma$ and $\delta$, then $\epsilon \in [\Gamma]^*$. We further denote by $[\Gamma]^\circ$ the result of deleting
from $[\Gamma]^*$ any clause involving an atom $r(x,y)$ or $r(y,x)$, where $r$  is an ordinary binary predicate. 
Notice, incidentally, that $[\Gamma]^\circ$ may feature ordinary binary predicates: however, all occurrences of these must be in atoms of the forms $r(x,x)$ or $r(y,y)$.  

This last observation prompts the introduction of some additional notation and terminology that will be used in the next lemma. Call a literal \textit{diagonal} if it is of the form $\pm r(u,u)$, where $r$ is a binary predicate and $u$ a variable. Let $\sigma$ be a relational signature and $\sigma' \subseteq \sigma$ such that $\sigma \setminus \sigma'$ consists only of binary predicates. A \textit{semi-diagonal 2-type over} $(\sigma, \sigma')$ is a maximal consistent set of literals over $\sigma$ each one of which is either a literal over $\sigma'$ or a diagonal literal. If $\fA$ is a 
structure interpreting $\sigma$ and $a$, $b$ distinct elements of the domain $\fA$, we denote by $\tp^\fA_{/\sigma'}[a,b]$ the unique semi-diagonal 2-type over $(\sigma, \sigma')$ satisfied by the pair $\langle a,b \rangle$. Thus, 
$\tp^\fA_{/\sigma'}[a,b]$ is just like $\tp^\fA[a,b]$, except that it is silent on the question of which binary relations
in $\sigma \setminus \sigma'$ are satisfied by the pairs  $\langle a,b \rangle$ and $\langle b,a \rangle$.

The following lemma, which
will form the core of our reduction of\linebreak
$\FinSat(\LtoPO)$ to $\FinSat(\LtoPOu)$, is, in effect, nothing more than the familiar completeness theorem for (ordered) propositional resolution.
\begin{lemma}
	Let $\Gamma$ be a set of clauses, over a signature $\sigma$, let $\sigma^-$ be the signature obtained by
	removing all the ordinary binary predicates from $\sigma$, and let $\tau^-$ be a semi-diagonal 2-type over $(\sigma, \sigma^-)$.  If $\models \tau^- \rightarrow [\Gamma]^\circ$, then there exists a 2-type $\tau$ over the signature $\sigma$ 
	such that $\models \tau \rightarrow \tau^-$ and $\models \tau \rightarrow \Gamma$.
	\label{lma:resolution}
\end{lemma}
\begin{proof}
	Enumerate the formulas of the forms $r(x,y)$ and $r(y,x)$, where $r$ is an ordinary binary predicate in $\sigma$, as 
	$\rho_1, \dots, \rho_n$. 
	Define a \textit{level-$i$ extension}
	of $\tau^-$ inductively as follows: (i) $\tau^-$ is a level-0 extension of $\tau^-$; (ii) if $\tau'$ is a level-$i$ extension of $\tau^-$ ($0 \leq i < n$), then $\tau' \wedge \rho_{i+1}$  and $\tau' \wedge \neg \rho_{i+1}$ are 
	level-$(i+1)$ extensions of $\tau^-$. Thus, the level-$n$ extensions of $\tau$ are exactly 
	the 2-types over $\sigma$ entailing
	$\tau^-$. 
	If $\tau'$ is a level-$i$ extension of $\tau^-$ ($0 \leq i < n$), we say that $\tau'$ {\em violates} a clause $\delta$
	if, for every literal in $\delta$, the opposite literal is in $\tau'$; we say that $\tau'$ {\em violates} a set of clauses
	$\Delta$ if $\tau'$ violates some $\delta \in \Delta$.
	Suppose now that $\tau'$ is a level-$i$ extension of $\tau^-$ ($0 \leq i < n$). We claim that, if both 
	$\tau' \wedge \rho_{i+1}$  and $\tau' \wedge \neg \rho_{i+1}$ violate $[\Gamma]^*$, then so does $\tau^-$. For otherwise, there must be a clause $\neg \rho_{i+1} \vee \gamma' \in [\Gamma]^*$ violated by $\tau' \wedge \rho_{i+1}$ and a clause $\rho_{i+1} \vee \delta' \in [\Gamma]^*$ violated by $\tau' \wedge \neg \rho_{i+1}$. But in that case $\tau'$ violates the ordinary binary resolvent $\gamma' \vee \delta'$,  contradicting the
	supposition that $\tau'$ does not violate $[\Gamma]^*$. This proves the claim. Now, since $\tau^-$ 
	by hypothesis entails $[\Gamma]^\circ$, it certainly does not violate $[\Gamma]^\circ$. Moreover, since it involves no atoms of the form $r(x,y)$ or $r(y,x)$ for $r$ an ordinary binary predicate, 
	$\tau^-$ does not violate $[\Gamma]^*$ either. By the above claim, then,
	there must be at least one level-$n$ extension $\tau$ of $\tau^-$ which does not
	violate $[\Gamma]^* \supseteq \Gamma$. Since $\tau$ is a 2-type, this proves the lemma.
\end{proof}

The next lemma allows us to eliminate atoms of the forms $r(x,y)$ and $r(y,x)$, where $r$ is an
ordinary binary predicate, from spread-form $\LtoPO$-formulas. 
Recall that, if $\Gamma$ is
a finite set of clauses, $\Gamma^{-1}$ denotes the result of transposing the variables $x$ and $y$ in $\Gamma$.
\begin{lemma}
Let $\phi$ be the spread-form $\LtoPO$-formula
\begin{equation*}
\begin{split}
\bigwedge_{\zeta \in Z} \exists x . \zeta \wedge &
 \forall x \forall y (x = y \vee \Gamma) \wedge\\ 
&       \bigwedge_{k=0}^2 \bigwedge_{h=0}^{m-1} \forall x \exists y (\lambda_k \rightarrow (\lambda_{\lfloor k+1 \rfloor}(y) \wedge \mu_h(y) \wedge \Delta_h)).
\end{split}
\end{equation*}
Here, $Z$ is a set of pure unary formulas;
$\lambda_0, \lambda_1, 
\lambda_2$ are mutually exclusive pure unary formulas;
$\mu_0, \dots, \mu_{m-1}$ are mutually exclusive pure unary formulas \textup{(}with $m \geq 1$\textup{)};
and $\Gamma, \Delta_1, \dots, \Delta_m$ are sets of clauses. Let $\phi^\circ$ be the corresponding formula
\begin{equation*}
\begin{split}
\bigwedge_{\zeta \in Z} \exists x . \zeta \wedge &
 \forall x \forall y (x = y \vee [\Gamma \cup \Gamma^{-1}]^\circ) \wedge\\ 
&       \bigwedge_{k=0}^2 \bigwedge_{h=0}^{m-1} \forall x \exists y (\lambda_k \rightarrow (\lambda_{\lfloor k+1 \rfloor}(y) \wedge \mu_h(y) \wedge [\Delta_h \cup \Gamma \cup \Gamma^{-1}]^\circ)). 
\end{split}
\end{equation*}
Then $\models \phi \rightarrow \phi^\circ$, and, moreover, if $\phi^\circ$ has a model over some domain $A$, then so has $\phi$.
\label{lma:spread}
\end{lemma}
\begin{proof}
It is immediate that $\models \phi \rightarrow \phi^\circ$, by the validity of resolution. 
Now suppose
$\fA$ is a structure such that $\fA \models \phi^\circ$; we define a structure $\fA'$ over the same domain as $\fA$, such that $\fA \models \phi$. 
Fix $a \in A$ and $h$ ($0 \leq h < m$). If $a$ satisfies one (hence: exactly one) of the formulas $\lambda_0$,
$\lambda_1$, $\lambda_2$, there exists $b$ such that $\fA \models \lambda_{\modt{k+1}}[b]$,
$\fA \models \mu_h[b]$ and
$\fA \models [\Delta_h \cup \Gamma \cup \Gamma^{-1}]^\circ[a,b]$. 
Let $\tau^- = \tp_{/\sigma'}^\fA[a,b]$
Since $[\Delta_h \cup \Gamma \cup \Gamma^{-1}]^\circ$ involves
no atoms of the forms $r(x,y)$ or $r(y,x)$, where $r$ is an ordinary binary predicate, 
we have $\models \tau^-  \rightarrow [\Delta_h \cup \Gamma \cup \Gamma^{-1}]^\circ$, and therefore, by Lemma~\ref{lma:resolution}, there is a 2-type $\tau$ such that $\models \tau \rightarrow \tau^-$ and $\models \tau \rightarrow
(\Delta_h \cup \Gamma \cup \Gamma^{-1})$. So set the interpretations of the ordinary binary predicates of $\phi$ such that
$\fA' \models \tau[a,b]$. Keeping $a$ fixed, carry out the above procedure for all values of $h$, thus
choosing $m$ witnesses for $a$. Since, in each case, the
chosen element $b$ satisfies $\mu_h$, these witnesses are all distinct, and so no clashes arise when setting 2-types in $\fA'$. Now 
carry out the above procedure for all values of $a$. If $a$ satisfies $\lambda_k$, then any $b$ chosen as a witness for $a$ satisfies $\lambda_{\modt{k+1}}$, so that $a$ could not previously have been chosen as a witness for $b$.  
Again, therefore, no clashes arise when setting 2-types in $\fA'$. At this stage, although $\fA'$ is not completely defined, we know that, however the construction of $\fA'$ is completed, for all $a \in A$ 
and $k$ ($0 \leq k < 3$) such that $\fA' \models \lambda_k[a]$,  and all $h$ ($0 \leq h < m$),
there will exist $b \in A \setminus \set{a}$ such that $\fA' \models (\Delta_h \wedge \Gamma \wedge \Gamma^{-1})[a,b]$.
Finally, suppose $a$, $b$ are distinct elements of $A$ 
for which $\tp^{\fA'}[a,b]$ has not yet been defined, and let $\tau = \tp_{/\sigma'}^{\fA}[a,b]$.
Since $\fA \models \psi$, it follows that $\models \tau^- \rightarrow [\Gamma \cup \Gamma^{-1}]^\circ$, and
hence by
Lemma~\ref{lma:resolution} that there exists a 2-type $\tau$ such that $\models \tau \rightarrow \tau^-$ and $\models \tau \rightarrow
(\Gamma \cup \Gamma^{-1})$. Again, set the interpretations of the ordinary binary predicates so that
$\fA' \models \tau[a,b]$; and repeat the process until $\fA'$ is completely defined. At the end of this process,
for any distinct $a$, $b$ of $A$, $\fA \models \psi$, $\tau \models \Gamma[a,b]$. Thus, $\fA' \models \phi$,
as required.
\end{proof}

\begin{theorem}
Let $\phi$ be an $\LtoPO$-formula in standard normal form with multiplicity $m$ over a signature $\sigma$. If $\phi$ has a finite model, then it has a model of size bounded by a doubly exponential function of $|\sigma| +m$.
Hence, any finitely satisfiable $\LtoPO$-formula $\phi$ has a model 
of size bounded by a doubly exponential function of
$\sizeOf{\phi}$, and so $\FinSat(\LtoPO)$ is in $\TwoNExpTime$.
\label{theo:mainPO}
\end{theorem}
\begin{proof}
We prove the first statement of the theorem. The remainder then follows by Lemma~\ref{lma:nfL2}.
By 
Lemma~\ref{lma:asymmetric}, let $\phi^*$ 
be an $\LtoPO$-formula in spread normal form~\eqref{eq:spreadNf} with multiplicity $3m$ over a signature $\sigma^*$ having the following properties:
\textup{(i)} $\models \phi^* \rightarrow \phi$; \textup{(ii)} if $\phi$ has a \textup{(}finite\textup{)} model then so has $\phi^*$; and \textup{(iii)} $|\sigma^*|$ is polynomially bounded as a function of  $|\sigma| + m$.
By rewriting the sub-formulas $\eta, \theta_0, \dots, \theta_{m-1}$ of $\phi^*$ in conjunctive normal form, we may assume that $\phi^*$ has the form required for Lemma~\ref{lma:spread}. This re-writing will not affect the signature or multiplicity of $\phi^*$.  By Lemma~\ref{lma:spread}, there is an $\LtoPOu$-formula
$\phi^\circ$ in weak normal form over the same signature as $\phi^*$, having the same multiplicity, and
satisfiable over the same domains, in which all occurrences of 
ordinary binary predicates are in atoms of the forms $r(x,x)$ or $r(y,y)$. Let $\phi'$ be
the result of replacing any such atoms in $\phi^\circ$ with the respective atoms $\hat{r}(x)$, $\hat{r}(y)$, where $\hat{r}$ is a fresh unary predicate for each ordinary binary predicate $r$. It is obvious that $\phi^\circ$ and $\phi'$ are satisfiable over the same domains. Moreover, given that the formulas 
$\lambda_0$, $\lambda_1$ and $\lambda_2$ are mutually exclusive, we may insert the condition $x \neq y$
in all $\forall \exists$-conjuncts of $\phi'$. Thus,
$\phi'$ is an $\LtoPOu$-formula in weak normal form over some signature $\sigma'$ with multiplicity $m'= 3m$ such that $|\sigma'|$ is polynomially bounded as function of $|\sigma| + m$. 
By Theorem~\ref{theo:mainPOUnary}, if $\phi'$ has a finite model, then it has a model of size bounded by a doubly exponential function of 
$|\sigma'| + m'$. Therefore, $\phi$ has a model of size bounded by a doubly exponential function of $|\sigma| + m$. 
\end{proof}

\section{Two-variable logic with one transitive\\ relation}
\label{sec:L21t}

The purpose of this section is to show that the logic $\LtoT$ has the triply 
exponential-sized finite model property (Theorem~\ref{theo:mainT}): if $\phi$ is a finitely satisfiable
$\LtoT$-formula, then $\phi$ has a model
of size bounded by some fixed triply exponential function of $\sizeOf{\phi}$. It follows that the finite satisfiability problem for $\LtoT$ is in \ThreeNExpTime. We proceed by reduction 
to the
corresponding problem for standard normal-form $\LtoPO$-formulas, but over signatures of  
exponential size, and with exponentially large multiplicities. 
Recall that, in $\LtoT$, we have a distinguished binary predicate, $\ft$, which must be interpreted as a transitive relation. When speaking about
2-types, we take the assumed transitivity of $\ft$ into account: specifically, if a 2-type
contains the literals $\ft(x,y)$ and $\ft(y,x)$, then it must also contain $\ft(x,x)$ and
$\ft(y,y)$. 

Let $A$ be a set and $T$ a transitive relation on $A$. A subset $B \subseteq A$ is {\em strongly connected} if, for
all distinct $a, b \in B$, $aTb$. It is obvious that the maximal strongly-connected
subsets of $A$ form a partition: we refer to the cells of this partition as
the $T$-{\em cliques} of $A$. If $C$ is a $T$-clique of $A$ and $|C| >1$, then $T \supseteq C \times C$; if, however, $C = \set{a}$, then $a$ may or may not be related to itself by $T$.
If $C$ and $D$ are distinct $T$-cliques of $A$, then we write: 
(i) $C <_T D$ if, for all $a \in C$ and $b \in D$,
$aTb$ but not $bTa$; (ii)
$C >_T D$ if, for all $a \in C$ and $b \in D$, 
$bTa$ but not $aTb$; and (iii)
$C \sim_T D$ if,  for all $a \in C$ and $b \in D$, neither $aTb$ nor $bTa$. It is routine to show:
\begin{lemma}
Let $A$ be a set and $T$ a transitive relation on $A$. Then the relation $<_T$ is a partial order on the set of $T$-cliques of $A$. Moreover, if $C$ and $D$ are distinct $T$-cliques, then
$C <_T D$ if and only if $D >_T C$ and, furthermore,
exactly one of $C <_T D$, $C >_T D$ and $C \sim_T D$ obtains.
\label{lma:po}
\end{lemma}
If $\fA$ is a structure interpreting a distinguished binary predicate $\ft$ as a transitive
relation over a domain $A$, we refer to the $\ft^\fA$-cliques, simply, as the {\em cliques} of $\fA$. 
We employ the following abbreviations:
\begin{align}
& \ft_\equiv (x,y) := \ft(x,y) \wedge \ft(y,x) \wedge x \neq y 
& & \ft_< (x,y) \equiv \ft(x,y) \wedge \neg \ft(y,x) \nonumber\\
& \ft_\sim (x,y) \equiv \neg \ft(x,y) \wedge \neg \ft(y,x) \wedge x \neq y 
& & \ft_> (x,y) \equiv \neg \ft(x,y) \wedge \ft(y,x). 
\label{eq:orderDefs}
\end{align}
It is then easy to see that the following validity holds:
\begin{equation}
\models \forall x \forall y (x = y \vee \ft_\equiv(x,y) \vee \ft_<(x,y) \vee \ft_>(x,y) \vee \ft_\sim(x,y)).
\label{eq:orderValidity}
\end{equation}

A formula of $\LtoT$ is said to be in {\em transitive normal form} if it conforms to the pattern
\begin{equation}
\begin{split}
\bigwedge_{\relVar \in \set{\equiv,<,>,\sim}} & \hspace{-6mm}
\forall x \forall y (\ft_{\relVar} (x,y) \rightarrow \eta_{\relVar}) 
\ \ \wedge\\
& \bigwedge_{h=0}^{m-1}   \bigwedge_{\relVar \in \set{\equiv,<,>,\sim}} \hspace{-6mm}
\forall x 
     \exists y (p_{h,\relVar}(x) \rightarrow (\ft_{\relVar} (x,y) \wedge \theta_{h,\relVar})).
\end{split}
\label{eq:nfT}
\end{equation}
where $m \geq 1$, the $p_{h, \relVar}$ are unary predicates, and the  $\eta_{\relVar}$ and $\theta_{h,\relVar}$ 
quantifier- and equality-free formulas not featuring either of the atoms $\ft(x,y)$ or $\ft(y,x)$.
Transitive normal form is---modulo trivial logical manipulation---a special case
of standard normal form~\eqref{eq:snf}. 
Note that, in the above definition,
the sub-formulas $\eta_{\relVar}$ and $\theta_{h,\relVar}$ may contain the atoms $\ft(x,x)$ or $\ft(y,y)$.

We have the following normal-form theorem for $\LtoT$.
\begin{lemma}
Let $\phi$ be a $\LtoT$-formula. There exists an $\LtoT$-formula $\phi^*$ in transitive normal form such that:
\textup{(i)} $\models \phi^* \rightarrow \phi$; \textup{(ii)} every model of $\phi$ 
can be expanded to a model
of $\phi^*$; and \textup{(iii)} $\sizeOf{\phi^*}$ is bounded by a polynomial function of $\sizeOf{\phi}$.
\label{lma:nfT}
\end{lemma}
\begin{proof}
By Lemma~\ref{lma:nfL2}, we may without loss of generality assume $\phi$ to be in standard normal form:
\begin{equation*}
\forall x \forall y (x = y \vee \eta) \wedge
\bigwedge_{h=0}^{m-1} \forall x \exists y (x \neq y \wedge \theta_{h}).
\end{equation*}
where $\eta, \theta_0, \dots, \theta_{m-1}$ are equality- and quantifier-free.
Suppose $\fA \models \phi$. 
For all $h$ ($0 \leq h < m$), all $\relVar \in \set{\equiv, <, >, \sim}$, let $p_{h,\relVar}$
be a  fresh unary predicate, and expand $\fA$ to an interpretation $\fA'$ by setting $\fA' \models p_{h,\relVar}[a]$
if there exists $b \in A \setminus \set{a}$ such that $\fA \models \ft_{\relVar} [a,b]$ and $\fA \models \theta_h[a,b]$. 
Further, set $\theta_{h,\relVar}$ to be the result of replacing all atoms of the forms $\ft(x,y)$ or $\ft(y,x)$ in $\theta_h$ by either $\top$ or $\bot$ as specified by $\ft_{\relVar}(x,y)$. Thus, 
setting $\omega$ to be the formula
\begin{equation*}
\bigwedge_{h=0}^{m-1}  \bigwedge_{\relVar \in \set{\equiv,<,>,\sim}} 
\forall x (p_{h,\relVar}(x) \rightarrow 
     \exists y (\ft_{\relVar}(x,y) \wedge \theta_{h,\relVar})),
\end{equation*}
we see by construction of $\fA'$ that $\fA' \models \omega$. Observe that 
none of the $\theta_{h,\relVar}$ contains either of the atoms $\ft(x,y)$ or $\ft(y,x)$.
Let $\psi'_1$ be the formula
\begin{equation*}
 \forall x \bigwedge_{h=0}^{m-1} \bigvee_{\relVar \in \set{\equiv, <, >, \sim}}  p_{h,\relVar}(x).
\end{equation*}
Since $\fA \models \bigwedge_{h=0}^{m-1}\forall x \exists y (x \neq y \wedge \theta_h)$, 
 and bearing in mind the validity~\eqref{eq:orderValidity},
it follows that $\fA' \models \psi'_1$. Moreover, under our general assumption that 
all domains have cardinality at least 2, $\psi'_1$ is logically equivalent to the formula
$\psi_1$ given by:
\begin{equation*}
 \forall x \forall y \left( x = y \vee  \left(\bigwedge_{h=0}^{m-1} \bigvee_{\relVar \in \set{\equiv, <, >, \sim}}  p_{h,\relVar}(x) \right)\right).
\end{equation*}
For all $\relVar \in \set{\equiv,<,>,\sim}$,
let $\eta_{\relVar}$ be the result of replacing all atoms of the forms $\ft(x,y)$ or $\ft(y,x)$ in $\eta$ by either $\top$ or $\bot$ as specified by $\ft_{\relVar}(x,y)$; and let $\psi_2$ be the formula
\begin{equation*}
\fA \models \bigwedge_{\relVar \in \set{\equiv,<,>,\sim}}
\forall x \forall y (\ft_{\relVar} (x,y) \rightarrow (x = y \vee \eta_\relVar)).
\end{equation*}
Since $\fA \models \forall x \forall y (x = y \vee \eta)$, we have $\fA' \models \psi_2$. Observe that 
none of the $\eta_{\relVar}$ contains either of the atoms $\ft(x,y)$ or $\ft(y,x)$.

Let $\phi^* = \psi_1 \wedge \psi_2 \wedge \omega$. Thus, $\sizeOf{\phi^*}$ is bounded by a polynomial function of $\sizeOf{\phi}$. We have shown that, if $\phi$ has a model, so does 
$\phi^*$. Moreover, it follows easily from~\eqref{eq:orderValidity} that $\models \phi^* \rightarrow \phi$.
\end{proof}

The following lemma, taken from~\cite{fo21t:KO05}, gives us a simple way to replace a collection $B$ of elements
in some structure $\fA$ interpreting a purely relational signature $\sigma$ with a `small' set of elements $B'$ in 
such a way that formulas of $\FOT$ do not notice the difference. %Where $\fA$ is clear from context, 
We employ the following notation where $\fA$ is a structure 
and $B, B' \subseteq A$. We denote the set of 1-types realized over $B$ by
$\tp^\fA[B]$. Likewise, we denote the set of 2-types realized by pairs of elements $\langle b, b' \rangle$, where $b \in B$ and $b' \in B'$, by $\tp^\fA[B,B']$. (There is no requirement that $B$ and $B'$ be disjoint.) When $B$ is a singleton, we write 
$\tp^\fA[b,B']$ in place of $\tp^\fA[\set{b},B']$.

\begin{lemma}[\cite{fo21t:KO05}, Prop.~4]
\label{lem:ssp} 
Let $\fA$ be a $\sigma$-structure not containing the distinguished predicate $\ft$, $B \subseteq A$, and $C := A \setminus B$. Then there is a
$\sigma$-structure $\fA'$ with domain $A' = B' \cup C$
for some set $B'$ of size exponential in $|\sigma|$,
such that
\begin{enumerate}[\textup{(}i\textup{)}]
\item %i
$\fA'_{|C} = \fA_{|C}$.
\item %ii
$\tp^{\fA'}[B'] = \tp^\fA[B]$, whence $\tp^{\fA'}[A'] = \tp^\fA[A]$;
\item %iii
$\tp^{\fA'}[B',B'] = \tp^{\fA}[B,B]$ and $\tp^{\fA'}[B', C] = \tp^{\fA}[B, C]$, whence\newline
$\tp^{\fA'}[A',A'] = \tp^{\fA}[A,A]$;
\item %iv
for each $b' \in B'$ there is some $b \in B$ with $\tp^{\fA'}[b',A']
\supseteq \tp^{\fA}[b,A]$;
\item %v
for each $a \in C$: $\tp^{\fA'}[a,B'] \supseteq \tp^{\fA}[a,B]$.
\end{enumerate}
\end{lemma}

The above Lemma applies to arbitrary structures (without any distinguished predicates). If, now, $\ft$ is a distinguished
predicate required to be interpreted as a transitive relation, let us write 
$\tp_<^\fA[B,B']$ to denote the subset of 2-types $\beta \in \tp^\fA[B,B']$ such that $\models \beta \rightarrow \ft_<(x,y)$, and similarly for $\tp_>^\fA[B,B']$ and $\tp_\sim^\fA[B,B']$. 
The following two Lemmas, due to ~\cite{fo21t:st13}, allow us to replace any clique in
a structure $\fA$ interpreting $\ft$ by an equivalent one of bounded size.
The proofs were kindly supplied by those authors in private
communication. 

\begin{lemma}\label{lem:sc}
Let $\sigma$ be a signature containing the distinguished transitive predicate $\ft$,
$\fA$ a $\sigma$-structure, and $\mathbb{A}$ the set of cliques of $\fA$. Let 
$B \in \mathbb{A}$ and $C = A \setminus B$.
Then there is a $\sigma$-structure $\fA'$ with domain
$A'= B' \cup C$ for some set $B'$,
with $|B'|$ bounded exponentially in $|\sigma|$, such that
\textup{(}i\textup{)}--\textup{(}v\textup{)} are as in Lemma \ref{lem:ssp}, and the set of cliques of $\fA'$ is $(\mathbb{A} \setminus \set{B}) \cup \set{B'}$.
\end{lemma}
\begin{proof}
If $|B|=1$, then we simply put $B'=B$ and we are done. Otherwise,
let $u$, $u_<$, $u_>$ and $u_\sim$ be fresh unary predicates.  Let
$\bar{\fA}$ be the expansion of $\fA$ obtained by setting $u^{\bar{\fA}} = B$ and 
\begin{equation*}
u_\relVar^{\bar{\fA}} = \set{a \in C : \quad  \fA \models \ft_\relVar[a,b] \text{ for some ($=$ all) $b \in B$}},
\end{equation*}
for $\relVar \in \set{<,>,\sim}$; and now rename the distinguished predicate $\ft$ in
$\bar{\fA}$ with an ordinary binary predicate---say---$q_0$. (Of course, even though $q_0$ is not a distinguished predicate, $q_0^{\bar{\fA}}$
is still a transitive relation.)
Let the result of applying Lemma \ref{lem:ssp} to
$\bar{\fA}$ and $B$ be a structure
$\bar{\fA}'$, in which $B'$ is the replacement for $B$; and write $A'$ for the domain of $\bar{\fA}'$. 
Notice that, if $\tau$ is a 2-type realized in $\bar{\fA}$ containing the literals $u_<(x)$ and $u(y)$,
then $\tau$ also contains the literals $q_0(x,y)$ and $\neg q_0(y)$, and similarly, {\em mutatis mutandis}, with $u_<$ replaced by $u_>$ and $u_\sim$.
But from property (iii) of Lemma \ref{lem:ssp}, we have $\tp^{\bar{\fA}'}[A',A'] = \tp^{\bar{\fA}}[A,A]$. Hence, if $a \in C$ is such that $\fA \models \ft_<[a,b]$ for some (and hence all) $b \in B$,
then $\bar{\fA}' \models q_0[a,b']$ and $\bar{\fA}' \not \models q_0[b',a]$ for 
all (and hence some) $b' \in B$, and similarly for $\ft_>$ and $\ft_\sim$.
It is then obvious that $q_0^{\bar{\fA}}$ is a transitive relation with set of cliques
$(\mathbb{A} \setminus \set{B}) \cup \set{B'}$, and indeed that the clique ordering induced by 
$\ft$
on $\fA$ and clique ordering induced by $q_0$ on $\bar{\fA}'$ are isomorphic under replacement of $B$ by $B'$. 
Now let $\fA'$
be the structure obtained from $\bar{\fA}'$ by
dropping the interpretations of $u, u_<, u_>$ and $u_\sim$ and renaming 
$q_0$ back to $\ft$.  
\end{proof}
Now for the promised lemma allowing us to confine attention to models with small cliques.
\begin{lemma} 
Let $\phi$ be a \textup{(}finitely\textup{)} satisfiable $\LtoT$-sentence in transitive normal form over a
signature $\sigma$. Then there exists a \textup{(}finite\textup{)} model of $\phi$ in which
the size of each clique is bounded exponentially in $|\sigma|$.
\label{lma:smallCliques}
\end{lemma}

\begin{proof}  
Let $\phi$ be as given in~\eqref{eq:nfT}, and suppose $\fA \models \phi$.
Let $B \subseteq A$ be a clique of $\fA$, let 
$C=A\setminus B$, and let $\fA'$, with domain $A'= B' \cup
C$, be the result of applying of Lemma~\ref{lem:sc} to $\fA$. 
We claim that
$\fA' \models \phi$.
The universally quantified conjuncts of $\phi$ are true in
$\fA'$ thanks to property (iii) of Lemma~\ref{lem:sc}. As for the
existential conjuncts, for
any $c \in C$, properties  (i) and (v) guarantee that $c$ has all
required witnesses. For any $b \in B'$, the same thing is
guaranteed by property (iv). This establishes the claim.

Now let $\fA$ be a countable $\sigma$-structure. Let $I_1$, $I_2,
\dots$ be a (possibly infinite) sequence of all cliques in a
$\fA$, $\fA_0=\fA$ and $\fA_{j+1}$ be the structure $\fA_j$
modified by replacing clique $I_{j+1}$ by its small replacement
$I'_{j+1}$ as described above. We define the limit structure
$\fA_\infty$ with the domain $I'_1 \cup I'_2, \dots$ such that
for all $k, l$ the connections between $I'_k$ and $I'_l$ are
defined in the same way as in $\fA_{max(k, l)}$. It is easy to see
that $\fA_\infty \models \phi$ and all cliques in $\fA_\infty$
are bounded exponentially in $|\sigma|$.
\end{proof}

We are now ready to prove the main result of this section: an exponential reduction of the (finite) satisfiability problem for $\LtoT$ to the (finite) satisfiability problem for standard normal form
$\LtoPO$-formulas.
\begin{lemma}
Let $\phi$ be formula of $\LtoT$. 
There exists an $\LtoPO$-formula $\hat{\phi}$ in standard normal form over a signature $\hat{\sigma}$
with multiplicity $\hat{m}$, such that:
\textup{(i)} if $\phi$ has a \textup{(}finite\textup{)} model with at least $2$ cliques, then
$\hat{\phi}$ has a \textup{(}finite\textup{)} model; \textup{(ii)}
if
$\hat{\phi}$ has a model of size $L$, then 
$\phi$ has a model of size at most $n \cdot L$, where $n$ is bounded by an exponential function of $\sizeOf{\phi}$; and \textup{(iii)} both $|\hat{\sigma}|$ and $\hat{m}$ are bounded by an exponential function of $\sizeOf{\phi}$.
\label{lma:cliquifiy}
\end{lemma}
\begin{proof}
By Lemma~\ref{lma:nfT}, we may without loss of generality assume $\phi$ to be in transitive normal form:
\begin{equation*}
\bigwedge_{\relVar \in \set{\equiv,<,>,\sim}} \hspace{-6mm}
\forall x \forall y (\ft_{\relVar} (x,y) \rightarrow \eta_{\relVar}) 
\wedge
\bigwedge_{h=0}^{m-1}   \bigwedge_{\relVar \in \set{\equiv,<,>,\sim}} \hspace{-6mm}
\forall x \exists y  (p_{h,\relVar}(x) \rightarrow 
     (\ft_{\relVar} (x,y) \wedge \theta_{h,\relVar})).
\end{equation*}
Let $\sigma$ be the signature of $\phi$.
From Lemma~\ref{lma:smallCliques}, we know that, if $\phi$ has a (finite) model, then it has
one in which each clique is of size at most $n$, where $n$ is bounded by an exponential function of $|\sigma|$.
Let $\bC = \set{c_1, \dots, c_n}$ be some set of $n$ objects ($n \geq 1$). We call any set $C = \set{c_1, \dots, c_m}$ for some $m$ ($1 \leq m \leq n$) an \textit{initial segment} of $\bC$. Say that a {\em cell} is a $\sigma$-structure
$\fC$ whose domain $C$ is an initial segment of $\bC$ such that $\fC$ has exactly one clique, namely $C$ itself. Here (and here only)
we lift our usual assumption that all structures have cardinality at least 2, thus allowing cells with the
singleton domain $\set{c_1}$.
Enumerate the cells as $\fC_0, \dots, \fC_{M-1}$. Thus, $M$ is bounded by a doubly exponential function of $|\sigma|$. 
Notice that, if $\fC$ is a cell containing more than 1 element, then 
by the transitivity of $\ft^\fC$, we have $\fC \models \forall x . \ft(x,x)$, and hence $\ft^\fC = C \times C$.
On the other hand, 
if $C = \set{c_1}$, then we may have either $\fC \models \ft[c_1, c_1]$ or $\fC \not \models \ft[c_1, c_1]$. (It follows, incidentally, that $M \geq 2$.) 

Now let $\bE = \set{e_1, \dots, e_n}$ and $\bE' = \set{e'_1, \dots, e'_n}$ be disjoint sets of cardinality $n$, and define the notion of
an \textit{initial segment} of these sets in the same way as for $\bC$. Say that a {\em diatom} is 
a $\sigma$-structure $\fD$ with domain $D = E \cup E'$, where $E$ is an initial segment of $\bE$ and $E'$ an initial segment of $\bE'$ (not necessarily of the same cardinality),
such that such that the set of cliques in $\fD$ is exactly $\set{E,E'}$. 
Enumerate the diatoms
as $\fD_0, \dots, \fD_{N-1}$. Thus, $N$ is bounded by a doubly exponential function of $|\sigma|$. 
(On the other hand, $N \geq M \geq 2$.)

If $C$ is an initial segment of $\bC$
we define the mappings $\epsilon: C \rightarrow \bE$ and $\epsilon': C \rightarrow \bE'$ by
$\epsilon(c_i) = e_i$ and $\epsilon'(c_i) = e'_i$ for all $i$ ($1 \leq i \leq |C|$). 
Thus, if $\fD = \fD_k$ is a diatom with cliques $E \subseteq \bE$  and $E' \subseteq \bE'$, there exist unique
cells $\fC = \fC_j$ and $\fC'= \fC_{j'}$ 
such that $\epsilon: \fC \simeq \fD_{|E}$ and $\epsilon': \fC' \simeq \fD_{|E'}$. 
We refer to $\fC$ and $\fC'$ as the \textit{left-} and \textit{right-cells} of $\fD$, respectively,
and, working with the corresponding indices, we define, for all $k$ ($0 \leq k < N$),
$\leftCell(k) = j$ and $\rightCell(k) = j'$.
Suppose now that we replaced the elements $e_1, e_2, \dots$
of $E \subseteq \bE$ with the corresponding elements $e'_1, e'_2, \dots$
of $\bE'$, and  we replaced the elements $e'_1, e'_2, \dots$
of $E' \subseteq \bE'$ with the corresponding elements $e_1, e_2, \dots$
of $\bE$. The result would be another diatom, say, $ \fD^{-1}$, obtained (in essence) by reversing the 
choice of which clique of $\fD$ defines the left-cell, and which the right-cell. 
We refer to $\fD^{-1} = \fD_{k'}$ as the {\em inverse} of $\fD = \fD_k$,
and, working with the corresponding indices, we define, for all $k$ ($0 \leq k < N$),
$\converseDiatom(k) = k'$. 
We introduce one final piece of terminology regarding diatoms. Recalling the abbreviations~\eqref{eq:orderDefs},
consider any diatom $\fD$ with cliques $E \subseteq \bE$ and $E' \subseteq \bE'$.
Evidently, for some $\relVar \in \set{<,>,\sim}$, we have $\fD \models \ft_\relVar[e,e']$ for all $e \in E$ and $e' \in E$. We call $\relVar$ the {\em order-type} of $\fD$. Thus, the order type of $\fD$ is $>$ if and only if
the order type of $\fD^{-1}$ is $<$, and the order type of $\fD$ is $\sim$ if and only if
the order type of $\fD^{-1}$ is $\sim$. Working with the corresponding indices, we denote the order-type of $\fD_k$ by $\relVar \langle k \rangle$.

Suppose $\phi$ has a model $\fA$ with at least two cliques, where no clique of $\fA$ has more than $n$ elements. We proceed to construct an $\LtoPO$-formula $\hat{\phi}$ (depending only on $\phi$, and not $\fA$), together with a model $\hat{\fA}$ of $\hat{\phi}$.
To avoid confusion, we
use the variables $u$ and $v$ in $\hat{\phi}$ in place of $x$ and $y$: it helps to think of $u$ and $v$ as
ranging over the set of cliques of $\fA$. 
Let $\bar{p} = p_1, \dots, p_s$ be a list of fresh unary predicates and $\bar{q} = q_1, \dots, q_t$ a list of fresh
binary predicates, where $\lceil s= \log M \rceil$ and $\lceil t= \log N \rceil$. Applying the same technique as employed in the proof of Lemma~\ref{lma:asymmetric}, we may form the labelling formulas $\bar{p}\langle j \rangle(u)$, for $0 \leq j < M$,
and $\bar{q}\langle k \rangle(u,v)$ for $0 \leq k < N$. 
Now let $\hat{A}$ be the set of cliques of $\fA$, and for each $\hat{a} \in \hat{A}$, fix some (arbitrary) 1--1 function $\hat{a} \rightarrow C$, where $C$ is the initial segment of $\bC$ of cardinality $|\hat{a}|$. 
Denote by $\iota: A \rightarrow \bC$ the union of all these functions.
(In effect, $\iota$ orders the elements in each cell.)
For any $\hat{a} \in \hat{A}$, the substructure $\fA_{|\hat{a}}$ is isomorphic, 
under $\iota$, to some cell or other, say, $\fC_{\hat{a}}$, which we call the 
\textit{reference cell} of $\hat{a}$. Now suppose that $\hat{a}, \hat{a}' \in \hat{A}$ are distinct, and let $\fC$ and $\fC'$ be their respective reference cells.
(There is no requirement that $\fC$ and $\fC'$ be distinct.) Recalling the functions  $\epsilon$
and $\epsilon'$ defined above, and setting $E = \epsilon(C)$, $E' = \epsilon'(C)$, 
define the function $\kappa: (\hat{a} \cup \hat{a}') \rightarrow E \cup E'$ (see Fig.~\ref{fig:reference}) by
\begin{equation*}
\kappa(a) = 
\begin{cases}
\epsilon(\iota(a)) & \text{ if $a \in \hat{a}$;}\\
\epsilon'(\iota(a)) & \text{ otherwise, (i.e.~if $a \in \hat{a}'$).}
\end{cases}
\end{equation*}
Evidently, $\kappa$ defines an isomorphism from $\fA_{|(\hat{a} \cup \hat{a}')}$ to some diatom or other, 
say $\fD_{\hat{a}, \hat{a}'}$, with cells $E$ and $E'$, which we call the \textit{reference diatom} of the
pair $\langle \hat{a}, \hat{a}' \rangle$.
Observe that $\fC_{\hat{a}}$ is always the left-cell of $\fD_{\hat{a}, \hat{a}'}$, and $\fC_{\hat{a}'}$ the right-cell.
That is, if $\fD_{\hat{a}, \hat{a}'} = \fD_k$, then $\fC_{\hat{a}} = \fC_{\leftCell(k)}$ and $\fC_{\hat{a}'} = \fC_{\rightCell(k)}$. Observe also that, if 
$\fD = \fD_{\hat{a}, \hat{a}'}$, then $\fD^{-1} = \fD_{\hat{a}', \hat{a}}$.
That is, if $\fD_{\hat{a}, \hat{a}'} = \fD_k$, then $\fD_{\hat{a}', \hat{a}}  = \fD_{I(k)}$.
\begin{figure}

\begin{center}
\begin{tikzpicture}[scale= 0.3]
% \fA
\draw (0, 15) -- (11,15) -- (11,2) --(0,2);
\draw (1,13) node {$\fA$};
\draw (5,4) rectangle (10,11);
\draw (6.2,6) node {$\hat{a}'$};
\draw (7,5) rectangle (9,7);
\draw (6,9) node {$\hat{a}$};
\draw (7,8) rectangle (9,10);
\draw[->-=0.6, bend right] (9,9) to node [above,pos=0.6] {$\iota$} (21,14);
\draw[->-=0.7, bend right, dashed] (9,9) to node [above,pos=0.7] {$\kappa$} (26,5);

\draw[->-=0.8, bend right] (9,6) to node [below,pos=0.8] {$\iota$} (31,14);
\draw[->-=0.7, bend right, dashed] (9,6) to node [above,pos=0.7] {$\kappa$} (26,2);

% \fD0
\draw (16, -1) node {$\fD_0$};
\draw (14,0) rectangle (18,7);

% \fD
%\draw (20, 3.5) node {$\dots$};
\draw (27, -1) node {$\fD_{\hat{a},\hat{a}'}$};
\draw (25,0) rectangle (29,7);
\draw (27,2) node {$E'$};
\draw (26,1) rectangle (28,3);
\draw (27,5) node {$E$};
\draw (26,4) rectangle (28,6);
%\draw (34, 3.5) node {$\dots$};

% \fDN-1
\draw (32, 3.5) node {$\dots$};
\draw (37, -1) node {$\fD_{N-1}$};
\draw (35,0) rectangle (39,7);
\draw (23, 3.5) node {$\dots$};

%\fC0
\draw (15, 17) node {$\fC_0$};
\draw (14,14) rectangle (16,16);

%\fC
\draw (18, 15) node {$\dots$};
\draw (21, 17) node {$\fC_{\hat{a}}$};
\draw (20,14) rectangle (22,16);
\draw[->-=0.3, bend right] (21,14) to node [right,pos=0.3] {$\epsilon$} (26,5);

% \fC'
\draw (26, 15) node {$\dots$};
\draw (31, 17) node {$\fC_{\hat{a}'}$};
\draw (30,14) rectangle (32,16);
\draw[->-=0.4, bend left] (31,14) to  node [right,pos=0.4] {$\epsilon'$} (28,2);
\draw (34.5, 15) node {$\dots$};

%\fCM-1
\draw (38, 17) node {$\fC_{M-1}$};
\draw (37,14) rectangle (39,16);

\draw[dotted] (13,-2) -- (13,18) -- (40,18) -- (40,-2) -- cycle;

\end{tikzpicture}
\end{center}
\caption{The function $\kappa$ mapping  $\hat{a} \cup \hat{a}'$ to 
the reference diatom $\fD = \fD_{\hat{a},\hat{a}'}$. The construction of $\kappa$ composes the function
$\iota$ mapping $\hat{a}$ and $\hat{a}'$
to their 
respective reference cells $\fC = \fC_{\hat{a}}$ and $\fC' = \fC_{\hat{a}'}$ with the functions $\epsilon: \bC \rightarrow \bE$ and
$\epsilon': \bC \rightarrow \bE'$.}
\label{fig:reference}
\end{figure}
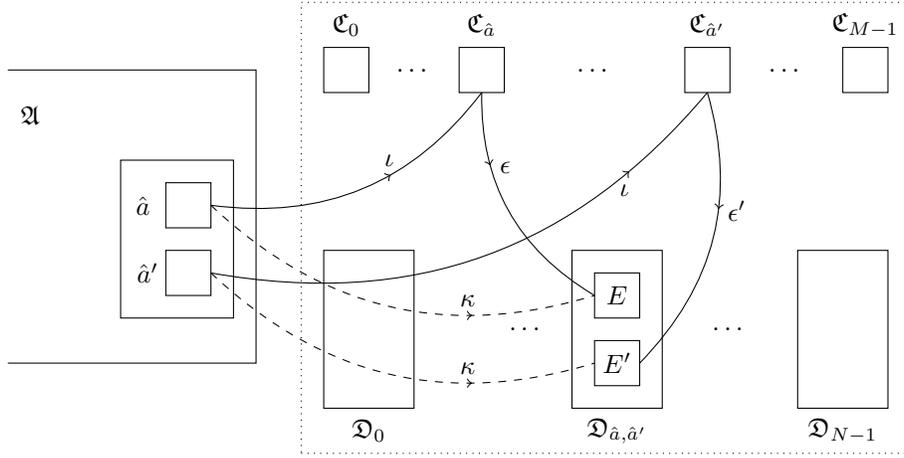

Now let $\hat{\fA}$ be the structure over $\hat{A}$ with
signature $\hat{\sigma} = \bar{p} \cup \bar{q} \cup \set{<}$, defined as follows.
\begin{enumerate}
\item For all $\hat{a} \in \hat{A}$,
and all $j$ ($0 \leq j < M$), $\hat{\fA} \models \bar{p}\langle j \rangle[\hat{a}]$ if and only
if $\fC_{\hat{a}} = \fC_j$;
\item
for all distinct $\hat{a}, \hat{a}' \in \hat{A}$
and all $k$ ($0 \leq k < N$), $\hat{\fA} \models \bar{q}\langle k \rangle[\hat{a},\hat{a}']$ if and only
if $\fD_{\hat{a},\hat{a}'}= \fD_k$;
\item
for all distinct $\hat{a}, \hat{a}'\in \hat{A}$, and all $\relVar \in \set{<,>,\sim}$,
$\hat{\fA} \models \relVar(\hat{a}, \hat{a}') $ if and only
if $\fD_{\hat{a}, \hat{a}'}$ is of order-type $\relVar$.
\end{enumerate}
Under this interpretation, and
taking the variables $u$ and $v$ to range over the
cliques of $\fA$, the formula $\bar{p}\langle j \rangle(u)$ says ``the reference cell  $\fC_u$ of $u$ is $\fC_j$,'' while 
the formula $\bar{q}\langle k \rangle(u,v)$ says ``the reference diatom $\fD_{u,v}$ of $\langle u, v \rangle$ is $\fD_k$.''  
Furthermore, by Lemma~\ref{lma:po}, $<^{\hat{\fA}}$ is a partial order, with
$>^{\hat{\fA}}$ and $\sim^{\hat{\fA}}$ standing in the expected relations to $<^{\hat{\fA}}$. 
Since $\fA$ by assumption has at least two cliques, $\hat{\fA}$ does not violate our general
assumption that all structures have cardinality at least 2. 

We now proceed to define the sought-after formula $\hat{\phi}$, building it up conjunct-by-conjunct, verifying, as we do so, that all these conjuncts are true in $\hat{\fA}$. 
We begin by taking $\psi'_1$ to be the conjunction
\begin{equation*}
\forall u \bigvee_{j=0}^{M-1} \bar{p}\langle j \rangle (u)
\wedge
\forall u \forall v \left( u = v \vee \bigvee_{k=0}^{N-1} \bar{q}\langle k \rangle (u,v) \right).
\end{equation*}
Under the interpretation $\hat{\fA}$, we may read $\psi'_1$ as saying: ``Every clique of $\fA$ has some reference cell, and every pair of distinct cliques has some reference diatom.'' This is obviously true by construction. Under the general assumption that all domains have cardinality at least 2, $\psi'_1$ is equivalent to the formula $\psi_1$ given by
\begin{equation*}
\forall u \forall v \left( u = v \vee \bigvee_{j=0}^{M-1} \bar{p}\langle j \rangle (u) \right)
\wedge
\forall u \forall v \left( u = v \vee \bigvee_{k=0}^{N-1} \bar{q}\langle k \rangle (u,v) \right),
\end{equation*}
so that $\hat{\fA} \models \psi_1$. Now let $\psi_2$ be the formula
\begin{equation*}
\bigwedge_{k=0}^{N-1}
\forall u \forall v \left(u = v \vee \left(\bar{q}\langle k \rangle(u,v) \rightarrow \bar{p}\langle \leftCell(k) \rangle(u) \wedge
\bar{p}\langle \rightCell(k) \rangle(v) \wedge \right)\right).
%\label{eq:noClash}
\end{equation*}
Under the interpretation $\hat{\fA}$, we may read $\psi_2$ as stating
that, if $u$ and $v$ are distinct cliques of $\fA$ such that $\fD_{u,v} = \fD_k$, then 
$\fC_{u} = \fC_{\leftCell(k)}$ and $\fC_{v} = \fC_{\rightCell(k)}$.
Now let $\psi_3$ be the formula
\begin{equation*}
\bigwedge_{k=0}^{N-1}
\forall u \forall v \left(u = v \vee \left(\bar{q}\langle k \rangle(u,v) \rightarrow \bar{q}\langle \converseDiatom(k) \rangle(u,v) \right)\right).
\end{equation*}
Under the interpretation $\hat{\fA}$, we may read $\psi_3$ as stating
that, if $u$ and $v$ are distinct cliques of $\fA$ such that $\fD_{u, v} = \fD_k$, then 
$\fD_{v,u} = \fD_{\converseDiatom(k)}$. 
Further, let $\psi_4$ be the formula
\begin{equation*}
\bigwedge_{k=0}^{N-1}
\forall u \forall v \left(u = v \vee \left(\bar{q}\langle k \rangle(u,v) \rightarrow \relVar \langle k \rangle (u,v) \right)\right).
\end{equation*}
Under the interpretation $\hat{\fA}$, we may read $\psi_4$ as stating
that, if $u$ and $v$ are distinct cliques of $\fA$ such that $\fD_{u, v} = \fD_k$, then the order-type of
$\fD_{v,u}$ is $\relVar \langle k \rangle$. Again, we have already observed that all these statements are true.
Thus, $\hat{\fA} \models \psi_2 \wedge \psi_3 \wedge \psi_4$. 

We now turn our attention to the formula $\phi$, starting with the purely 
universal conjuncts. Let $\lambda(u)$ abbreviate the formula
\begin{equation*}
\bigvee \set{\bar{p}\langle j \rangle (u) \mid 0 \leq j < M,\ \fC_j \models \forall x \forall y 
   (x= y \vee \eta_{\equiv}(x,y))}.
\end{equation*}
Under the interpretation $\hat{\fA}$, we may read $\lambda(u)$ as ``$u$ is a clique of $\fA$ in which 
the formula $\eta_{\equiv}(x,y)$ is satisfied by all pairs of distinct elements.''
For each $\relVar \in \set{<,>,\sim}$, 
let $\hat{\eta}_{\relVar}(u,v)$ abbreviate the formula
\begin{equation*}
\bigvee \set{\bar{q}\langle k \rangle (u,v) \mid 0 \leq k < N,\ \fD_k \models \forall x \forall y 
   (\ft_{\relVar}(x,y) \rightarrow \eta_{\relVar}(x,y))}.
\end{equation*}
We may read $\hat{\eta}_{\relVar}(u,v)$ as ``$u$ and $v$ are a pair of cliques in which 
the formula $\eta_{\relVar}(x,y)$ is satisfied by all pairs of elements related by $\ft_\relVar$.''
Now let $\psi'_5$ be the formula 
\begin{equation*}
\forall u .\lambda(u) \ \wedge \bigwedge_{\relVar \in \set{<,>,\sim}} \forall u \forall v (x = y \vee \hat{\eta}_{\relVar}(u,v)).
\end{equation*}
Under the interpretation $\hat{\fA}$, we may read the first conjunct of $\psi'_5$ as stating: 
``if $\fC$ is a cell realized in $\fA$, then any pair of distinct elements in $\fC$ satisfies
$\eta_\equiv(x,y)$.'' The truth of this statement follows from the fact that $\fA \models \forall x \forall x (\ft_\equiv(x,y) \rightarrow \eta_\equiv(x,y))$.
Similarly, the remaining conjuncts state: ``if $\fD$ is a diatom realized in $\fA$ having order-type $\relVar$ then any pair of elements ordered by $\ft_s$ satisfies $\eta_\relVar$.'' The truth of this statement follows from the fact that $\fA \models \forall x \forall x (\ft_{\relVar}(x,y) \rightarrow \eta_{\relVar}(x,y))$.
Again, replacing $\psi'_5$ with the equivalent formula $\psi_5$ given by 
\begin{equation*}
\forall u \forall v (u = v \vee \lambda(u)) \ \wedge \bigwedge_{\relVar \in \set{<,>,\sim}} \forall u \forall v (x = y \vee  \hat{\eta}_{\relVar}(u,v)),
\end{equation*}
we see that, $\hat{\fA} \models \psi_5$.

Now we turn our attention to the universal-existential conjuncts of $\phi$.
For each $h$ ($0 \leq h < m$),
let $\mu_{h}(u)$ abbreviate the formula
\begin{equation*}
\bigvee \set{\bar{p}\langle j \rangle (u) \mid 0 \leq j < M,\ \fC_j \models \forall x (p_{h,\equiv}(x) \rightarrow \exists y (x \neq y \wedge \theta_{h,\equiv}(x,y)))}.
\end{equation*}
Under the interpretation $\hat{\fA}$, we may read $\mu_{h}(u)$ as ``$u$ is a clique isomorphic to some cell $\fC$ such that $\fC \models \forall x (p_{h,\equiv}(x) \rightarrow \exists y (x \neq y \wedge \theta_{h,\equiv}(x,y)))$.'' 
Now let $\psi'_6$ be the formula 
\begin{equation*}
\forall u \bigwedge_{h=0}^{m-1} \mu_h(u). 
\end{equation*}
Under the interpretation $\hat{\fA}$,
we may read $\psi'_6$ as stating: 
``if $\fC$ is a cell realized in $\fA$, then any element in $\fC$ satisfying $p_{h,\equiv}(x)$ has a witness for 
$\exists y(x \neq y \wedge \theta_{h,\equiv}(x,y))$ in $\fC$.
The truth of this statement follows from the fact that $\fA \models \forall x (p_{h,\equiv}(x) \rightarrow \exists y (\ft_\equiv(x,y) \wedge  \theta_{h,\equiv}(x,y)))$. Replacing $\psi'_6$ with $\psi_6$, given by
\begin{equation*}
\forall u \forall v \left( u = v \vee \bigwedge_{h=0}^{m-1}\mu_h(u) \right),
%\label{eq:secureAlpha}
\end{equation*}
we thus have $\hat{\fA} \models \psi_6$.
Further,
for each $h$ ($0 \leq h < m$), each $\relVar \in \set{<,>,\sim}$ and each $i$ ($0 \leq i <n$),
let $\nu_{h,\relVar,i}(u)$ abbreviate the formula
\begin{equation*}
\bigvee \set{\bar{p}\langle j \rangle (u) \mid 0 \leq j < M,\ \fC_j \models p_{h,s}[c_i]}.
\end{equation*}
We may read $\nu_{h,\relVar,i}(u)$ as ``$u$ is a clique whose reference cell $\fC_u$ is such that $\fC_u \models p_{h,s}[c_i]$.''
(We take the statement ``$\fC_u \models p_{h,s}[c_i]$'' to be false if $c_i$ is not in the domain of $\fC_u$.)
Finally, for each $h$ ($0 \leq h < m$), each $\relVar \in \set{<,>,\sim}$, each $i$ ($0 \leq i <n$) and each $i'$ ($0 \leq i' <n$),
let $\xi_{h,\relVar,i,i'}(u,v)$ be the formula
\begin{equation*}
\bigvee \set{\bar{q}\langle k \rangle (u,v) \mid 0 \leq k < N,\ \relVar\langle k \rangle = \relVar \text{ and } \fD_k \models \theta_{h,\relVar}[c_i,c'_{i'}]}.
\end{equation*}
We may read $\xi_{h,\relVar,i,i'}(u,v)$ as ``$u$ and $v$ are cliques whose reference diatom $\fD_{u,v}$
has order-type $\relVar$ and is such that $\fD_{u,v} \models \theta_{h,\relVar}[c_i,c'_{i'}]$.''
(We take the statement ``$\fD_{u,v} \models \theta_{h,\relVar}[c_i,c'_{i'}]$'' to be false if $c_i$
or $c_i'$ are not in the domain of $\fD_{u,v}$.) Now let $\omega$ be the conjunction
\begin{equation*}
\bigwedge_{\relVar \in \set{<,>,\sim}} \bigwedge_{i=0}^{n-1}\bigwedge_{h=0}^{m-1} \hspace{-0mm}
\forall u \exists v \left(u \neq v \wedge 
     \left(\nu_{h,\relVar,i}(u) \rightarrow 
           \bigvee_{i'=0}^{n-1} \xi_{h,\relVar,i,i'}(u,v)\right) \right).
\end{equation*}
Under the interpretation $\hat{\fA}$, 
$\omega$ states: ``for all $\relVar$ and $h$, if $u$ is an $\fA$-clique with reference cell $\fC$ such that some element $a$ of $C$ satisfies $p_{h,\relVar}(x)$, then there is some other $\fA$-clique $v$ of such that 
$a$ has a witness for $\exists y(\ft_\relVar(x,y) \wedge \theta_{h,\relVar}(x,y))$ in 
$\fA_{|(u \cup v)}$.'' The truth of this statement follows from the fact that $\fA \models \forall x (p_{h,\relVar}(x) \rightarrow \exists y (\ft_\relVar(x,y) \wedge  \theta_{h,\relVar}(x,y)))$.

Now let $\hat{\phi} = \psi_1 \wedge \cdots \wedge \psi_6 \wedge \omega$. Thus,
$\hat{\phi}$ is an $\LtoPO$-formula in standard normal form over a signature $\hat{\sigma}$ consisting of the unary predicates
$p_1, \dots, p_s$, the ordinary binary predicates $q_1, \dots, q_t$ and the navigational predicates 
$<$, $>$ and $\sim$,
with multiplicity $\hat{m} = 4mn$. We see that both $|\hat{\sigma}|$ and $\hat{m}$ are bounded by an exponential function of $\sizeOf{\phi}$.
Moreover, we have shown that $\hat{\phi}$ has the model $\hat{\fA}$, where $<$ is interpreted as the partial order $<_T$ on the cliques of $\fA$, and $>$ and $\sim$ stand in the usual relations to $<$. 
It is obvious that $\hat{\fA}$ is finite if $\fA$ is. This establishes conditions (i) and (iii) of the lemma. 

To establish condition (ii), we show that, if $\hat{\phi}$ has a model of size $L \geq 2$, then 
$\phi$ has a model of size at most $n \cdot L$. 
Suppose then that $\fB \models \hat{\phi}$, with $|B| = L$.
Consider any element $b \in B$. From~$\psi_1$ there exists $j$ ($0 \leq j < M$) such that $\fB \models \bar{p}\langle j \rangle[b]$, so let
$\fC_b$ be a fresh copy of the cell $\fC_j$, having domain, say, $\check{B}_b$. Let $\check{B} = \bigcup_{b \in B} \check{B}_b$,
and define a structure $\check{\fB}$ over $\check{B}$ as follows. For all $b \in B$, let  $\check{\fB}_{|\check{B}_b}= \fC_b$, so that it remains only to define the 2-types involving elements from different
sets $\check{B}_b$. Suppose $b, c \in B$ are distinct. From~$\psi_1$ again,
there exists $k$ ($0 \leq k < N$) such that $\fB \models \bar{q}\langle k \rangle[b,c]$. 
Now set $\check{B}_{|(\check{B}_b \cup \check{B}_c)}= \fD_k$. That these assignments do not clash with the
structures $\fC_b$ already established
is immediate from~$\psi_2$. 
That these assignments do not clash with each other
is immediate from~$\psi_3$. 
This completes the construction of $\check{\fB}$. Obviously, $|\check{B}| \leq n \cdot L$.
From~$\psi_5$, we have 
\begin{equation*}
\check{\fB} \models \bigwedge_{\relVar \in \set{\equiv,<,>,\sim}}
\forall x \forall y (\ft_{\relVar} (x,y) \rightarrow (x = y \vee \eta_{\relVar})).
\end{equation*}
Likewise, from~$\psi_6 \wedge \omega$, we have 
\begin{equation*}
\check{\fB} \models \bigwedge_{h=0}^{m-1}   \bigwedge_{\relVar \in \set{\equiv,<,>,\sim}}
\forall x (p_{h,\relVar}(x) \rightarrow 
     \exists y (\ft_{\relVar} (x,y) \wedge \theta_{h,\relVar})).
\end{equation*}
That is, $\check{\fB} \models \phi$, as required. It remains only to check that $\ft^{\check{\fB}}$ is transitive. 
By assumption, $<^\fB$ is a partial order. By construction, if $a$ and $a'$ are distinct
elements of $\check{B}_b$, for some $b \in B$, then $\check{\fB} \models \ft[a,a']$. From $\psi_4$, 
if $a \in \check{B}_b$ and $a' \in \check{B}_{b'}$, where
$b$ and $b'$ are distinct elements of $B$, then $\check{\fB} \models \ft[a,a']$ if and only
if $\fB \models b < b'$. It is then obvious that $\ft^{\check{\fB}}$ is transitive.
\end{proof}

\begin{theorem}
Any finitely satisfiable $\LtoT$-formula $\phi$ has a model 
of size bounded by a triply exponential function of
$\sizeOf{\phi}$, and so $\FinSat(\LtoT)$ is in $\ThreeNExpTime$.
\label{theo:mainT}
\end{theorem}
\begin{proof}
Let $\phi$ be a formula of $\LtoT$. 
Recalling our general assumption that all structures have cardinality at least 2, 
any model of $\phi$ consisting of a single clique is one in which $\ft$ is total. Thus, we may test
satisfiability of $\phi$ in single-clique structures by replacing all $\ft$-atoms by $\top$, and considering the
resulting $\FOT$-formula. Since any satisfiable $\FOT$-formula $\phi'$ has a model of cardinality bounded by an exponential function of $\sizeOf{\phi'}$, the result is established. Thus, we may confine our attention to determining whether $\phi$ has a finite model with at least 2 cliques.

By Lemma~\ref{lma:cliquifiy},
let $\hat{\phi}$ 
be an $\LtoPO$-formula in standard normal form with multiplicity $\hat{m}$ over a signature $\hat{\sigma}$, such that: (i) if $\phi$ has a finite model with at least 2 cliques, then $\hat{\phi}$ is finitely satisfiable, 
(ii) if $\hat{\phi}$ has a model of size $L$, then $\phi$ has a model of size $n \cdot L$,
where $n$ is bounded by an exponential function of $\sizeOf{\phi}$; and (iii) both $|\hat{\sigma}|$ and $\hat{m}$ are bounded by an exponential function of $\sizeOf{\phi}$.
By Theorem~\ref{theo:mainPO}, if $\phi$, and therefore
$\hat{\phi}$, is finitely satisfiable, then $\hat{\phi}$ has a model of size $L$ bounded by a
doubly exponential function of $|\hat{\sigma}| +\hat{m}$, and hence by a triply exponential function of $\sizeOf{\phi}$,
whence $\phi$ also has a model of size bounded by a triply exponential function of $\sizeOf{\phi}$.
\end{proof}

% Bibliography in this style
\bibliographystyle{plain}
\bibliography{fo21t}
%\tableofcontents
\end{document}